\documentclass{article}

\usepackage{microtype}
\usepackage{graphicx}
\usepackage{subfigure}
\usepackage{booktabs} %

\usepackage[backref=page]{hyperref}

\usepackage[accepted]{icml2023}

\usepackage{amsmath}
\usepackage{amssymb}
\usepackage{mathtools}
\usepackage{amsthm}

\usepackage{times}
\usepackage{ifthen}
\usepackage{tikz}
\usepackage{microtype}
\usepackage{bm}
\usepackage{amsmath}
\usepackage{nicefrac}
\usepackage{xpatch}
\usepackage[capitalize,noabbrev]{cleveref}
\usepackage{booktabs}
\usepackage[inline]{enumitem}
\usepackage[noend]{algpseudocode}
\usepackage{tikz,tikz-3dplot}
\usetikzlibrary{3d,arrows.meta,patterns}
\tdplotsetmaincoords{120}{60}

\usetikzlibrary{cd,arrows,calc,shapes,positioning}
\tikzstyle{obs} = [circle,fill=white,draw=black,inner sep=1pt,minimum size=20pt,font=\fontsize{10}{10}\selectfont,node distance=1,thick]
\tikzstyle{latent} = [obs,dotted]

\newcommand{\bR}{\ensuremath \mathbb{R}}
\newcommand{\bN}{\ensuremath \mathbb{N}}

\newcommand{\bE}{\ensuremath \mathbb{E}}

\newcommand{\cC}{\ensuremath \mathcal{C}}
\newcommand{\cI}{\ensuremath \mathcal{I}}
\newcommand{\cJ}{\ensuremath \mathcal{J}}

\newcommand{\cN}{\ensuremath \mathcal{N}}

\DeclareMathOperator*{\argmax}{arg\,max}
\DeclareMathOperator*{\argmin}{arg\,min}
\DeclareMathOperator{\E}{\mathbb{E}}
\DeclareMathOperator{\Var}{\mathrm{Var}}
\DeclareMathOperator{\Cov}{\mathrm{Cov}}
\DeclareMathOperator{\indep}{\perp\!\!\!\perp}
\DeclareMathOperator{\dep}{\not \perp\!\!\!\perp}
\DeclareMathOperator{\rk}{\mathrm{rk}}
\newcommand{\B}[1]{\bm{#1}}
\newcommand{\given}{\,|\,}
\DeclareMathOperator{\im}{im}

\newcommand{\ninstruments}{\ensuremath N_{\mathrm{IV}}}
\newcommand{\nperround}{\ensuremath N_{\nicefrac{\mathrm{IV}}{\mathrm{exp}}}}
\newcommand{\nexperiments}{\ensuremath T}

\newcommand{\X}{\ensuremath \B{X}}
\newcommand{\M}{\ensuremath \B{M}}
\newcommand{\U}{\ensuremath \B{U}}
\newcommand{\V}{\ensuremath \B{V}}
\newcommand{\y}{\ensuremath \B{y}}
\newcommand{\Z}{\ensuremath \B{Z}}
\renewcommand{\P}{\ensuremath \B{P}}
\newcommand{\D}{\ensuremath \B{D}}
\newcommand{\betaiv}{\ensuremath \hat{\beta}_{\mathrm{2SLS}}}
\newcommand{\betahat}{\ensuremath \widehat{\|\beta\|}_2}

\newcommand{\dimX}{{d_x}}
\newcommand{\dimZ}{{d_z}}

\newcommand{\cost}{c}
\newcommand{\similarity}{\mathrm{sim}}
\newcommand{\gain}{\mathrm{gain}}
\newcommand{\score}{\mathrm{score}}

\newcommand{\betaivsub}[1]{\ensuremath \widehat{\P_{\alpha_{#1}}\beta}}
\newcommand{\xhdr}[1]{\noindent \textbf{#1.}\;}

\theoremstyle{plain}
\newtheorem{theorem}{Theorem}
\newtheorem{proposition}[theorem]{Proposition}

\newtheorem{corollary}[theorem]{Corollary}
\theoremstyle{definition}

\theoremstyle{remark}

\makeatletter
\xpatchcmd{\algorithmic}{\itemsep\z@}{\itemsep=0.5ex plus1pt}{}{}
\makeatother

\icmltitlerunning{\hfill Sequential Underspecified Instrument Selection for Cause-Effect Estimation\hfill \thepage}

\begin{document}

\twocolumn[
\icmltitle{Sequential Underspecified Instrument Selection\\ for Cause-Effect Estimation}

\icmlsetsymbol{equal}{*}

\begin{icmlauthorlist}
\icmlauthor{Elisabeth Ailer}{sch,tum}
\icmlauthor{Jason Hartford}{yyy,comp}
\icmlauthor{Niki Kilbertus}{sch,tum}
\end{icmlauthorlist}

\icmlaffiliation{yyy}{MILA - Quebec Artificial Intelligence Institute, Montreal, Quebec, Canada}
\icmlaffiliation{sch}{HelmholtzAI, Helmholtz Munich, Munich, Germany}
\icmlaffiliation{comp}{Recursion, Montreal Quebec Canada}
\icmlaffiliation{tum}{TUM School of Computation, Information and Technology, Technical University Munich, Munich, Germany}

\icmlcorrespondingauthor{Elisabeth Ailer}{elisabeth.ailer@helmholtz-munich.de}

\icmlkeywords{instrumental variables, underspecified, experiment design, cause-effect, identification}

\vskip 0.3in
]

\printAffiliationsAndNotice{}  %

\begin{abstract}
Instrumental variable (IV) methods are used to estimate causal effects in settings with unobserved confounding, where we cannot directly experiment on the treatment variable. Instruments are variables which only affect the outcome indirectly via the treatment variable(s). Most IV applications focus on low-dimensional treatments and crucially require at least as many instruments as treatments. This assumption is restrictive: in the natural sciences we often seek to infer causal effects of high-dimensional treatments (e.g., the effect of gene expressions or microbiota on health and disease), but can only run few experiments with a limited number of instruments (e.g., drugs or antibiotics). In such underspecified problems, the full treatment effect is not identifiable in a single experiment even in the linear case. We show that one can still reliably recover the projection of the treatment effect onto the instrumented subspace and develop techniques to consistently combine such partial estimates from different sets of instruments. We then leverage our combined estimators in an algorithm that iteratively proposes the most informative instruments at each round of experimentation to maximize the overall information about the full causal effect.
\end{abstract}

\section{Introduction}

\xhdr{Motivation}
Understanding cause-effect relationships in high-dimensional systems is a common challenge in various scientific areas.
For example, how does our gut microbiome (treatment $X$) causally influence health and disease (outcome $Y$)? How does the transcriptome of a cell (treatment $X$) causally influence its function (outcome $Y$)?
Typically, the high-dimensional treatment and the outcome are heavily confounded via unknown mechanisms, rendering the strong assumptions required for cause-effect estimation, i.e., computing $p(y \given do(x))$, from observational data sampled from $p(x, y)$ untenable.
Hence, experimentation is indispensable for the ultimate goal of identifying and estimating these causal relationships.
Whenever we can intervene directly on the treatment, we have direct access to $p(y \given do(x))$ and cause-effect relationships can be captured by mere association in experiments.
We are motivated by two crucial realizations: (a) oftentimes practically feasible experiments do not intervene directly on the treatment $X$ but some other variable $Z$; (b) still, the scientific goal is to estimate the effect of the high-dimensional $X$ on the outcome $Y$ (instead of the effect of the actual intervention on $Z$).

Regarding (a), administering certain antibiotics has a strong and highly predictable effect on the gut microbiome, but it does not break causal links from potentially unobserved confounders to $X$; similarly, applying various drug (dosages) to cell cultures influences the transcriptome, but again does not directly intervene on it.
Even targeted gene knockouts or CRISPR/Cas gene editing \citep{zhang2015crispr, fu2013crispr} do not constitute interventions as defined in causal inference \citep{pearl2009causality}.
They do not strip the microbiome or transcriptome free of any other causal influences, i.e., they may still be confounded with the outcome of interest.
As for (b), ultimately we give a certain antibiotic to learn about \emph{how the microbiome causally influences disease} ($p(y \given do(x))$) and not merely about what overall effect the antibiotic has on disease ($p(y \given do(z))$).

In such settings, the variable experimented on (antibiotics, drugs, gene knockout/edits) can at most serve as an \emph{instrument} $Z$ for the treatment:
\textbf{(A1)} $Z$ (strongly) affects the treatment ($Z \dep X$).
\textbf{(A2)} $Z$ is independent of unobserved confounders $U$ ($Z \indep U$).
\textbf{(A3)} $Z$ ``only affects the outcome via the treatment'' ($Y \indep Z \given \{X, U\}$).
Whether we can ascertain these conditions depends on the setting.
For example, orally administered sub therapeutic dosages of antibiotics (such that no antibiotics can be detected in the bloodstream) may satisfy this condition \citep{ailer2021causal}.
Similarly, genetic interventions (or drugs) may indeed only influence cell function via the transcriptome.

As a result, even with access to experimentation, we must resort to instrumental variable (IV) techniques to estimate the treatment effect.
A number of challenges arise in this setting even in the fully linear case:
\textbf{(B1)} At least as many instruments as treatment variables are required for the causal effect to be identified (just- or overspecified setting).
This is not feasible in our scenarios with (tens of) thousands of organisms/genes and much fewer antibiotics/drugs to experiment with.
\textbf{(B2)} Experimentation is typically slow and expensive, putting natural bounds on how many experiments can be run.
\textbf{(B3)} The cost of an individual experiment depends on the number of randomized instruments and there may be limits to how many instruments can be sensibly combined in one single experiment. For example, too many drugs, gene knockouts, or antibiotics at once may kill or permanently damage the studied organism.

\xhdr{Main goal and contributions} Constrained by (B1)-(B3), we formulate our main goal:
\emph{Can we select a bounded number of instruments in a bounded number of sequential experiments and combine the results for an informative estimate of a high-dimensional causal effect?}\footnote{We highlight that while we call the high-dimensional $X$ the ``treatment'', we experiment (or intervene) on lower-dimensional instruments $Z$.}
In answering this question, we make the following contributions:

\begin{itemize}[leftmargin=*,itemsep=-2pt,topsep=0pt]
  \item We show that in the underspecified linear IV setting, we can estimate the orthogonal projection of the causal effect onto the ``instrumented subspace'' and construct a $\sqrt{n}$-consistent, asymptotically normal estimator.
  \item We combine such estimates from experiments with different instruments to consistently recover the estimate had we randomized all IVs simultaneously (without actually having to apply all perturbations at once).
  \item We develop an algorithm that sequentially proposes subsets of the available instruments to maximally identify the treatment effect from the combined estimate across all experimental rounds.
  The algorithm trades off the information gained from multiple instruments with the increasing cost of including them in a single experiment based on a pre-specified similarity between instruments.
  \item We develop techniques to keep track of which components of $\beta$ have been identified reliably at each round, upper bound the absolute error in unidentified components, and propose a stopping criterion based purely on observational data ($p(x,y)$) that---when reached---guarantees full identification of $\beta$ under mild assumptions.
\end{itemize}

\xhdr{Related work}
We build on the theory of (linear) instrumental variable estimators, with a specific focus on two-stage methods \citep{angrist2008mostly}.
Instrumental variables have been used since 1928 by Philip G.~Wright \citep{wright1928tariff,stock2003retrospectives} and are an essential part of the econometrics toolkit \citep{angrist2008mostly}.
They have also received renewed interest from the machine learning community recently \citep{hartford2017deep,zhang2020maximum,kilbertus2020class,singh2019kernel, bennett2019deep,Padh2022Bounding,saengkyongam2022exploiting,muandet2019dual}.
Despite the difficulty of finding valid instruments for a given target effect in practice \citep{hernan2006instruments}, instrumental variable estimation has been applied successfully in genetics via Mendelian randomization \citep{sanderson2022mendelianrandomization,didelez2007mendelianrandomization} and recently on microbiome data \citep{Sohn2019,Wang2020,ailer2021causal}.

Our work is also related to ideas in experiment design for causal structure learning \citep{hyttinen2013experiment,gamella2020active,sussex2021near,tigas2022interventions}.
Two key differences are that those works focus on sequential selection of \emph{interventions} (not just instruments) and that they seek to identify causal structure (i.e., the causal graph) instead of a specific causal effect.
To the best of our knowledge, there is no literature on adaptively selecting instruments in sequential experiments to identify a causal effect from high-dimensional treatments.

Finally, variants of underspecified instrumental variable settings have been studied by \citet{pfister2022identifiability}.
They assume the causal effect from $X$ on $Y$ to be sparse, which allows them to relax standard identifiability assumptions in the linear IV setting.
\citet{rothenhausler2018anchor} make use of exogenous variables to provide an estimator that interpolates between the ordinary least-squares (OLS) and two stage least-squares (2SLS) estimates, even when IV assumptions are not fully satisfied and point-identification is not guaranteed.
Their proposal can be interpreted as ``choosing the best performing (in terms of mean squared error) estimate among the compatible ones''.
Thus, both works pursue goals orthogonal to ours.

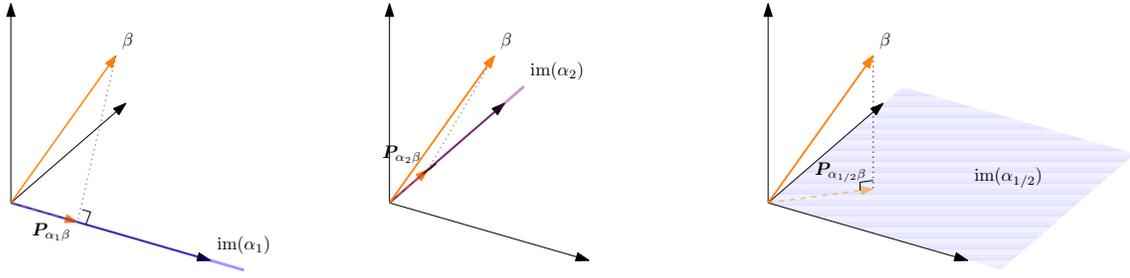
\begin{figure*}
\hfill
\begin{minipage}[b]{0.4\linewidth}
    \centering
    \resizebox{\linewidth}{!}{
        \begin{tikzpicture}[>={Stealth[inset=0pt,length=8pt,angle'=35,round]},tdplot_main_coords,scale=0.75]
        \draw[draw=white,fill=blue,fill opacity=0.0] (-4,-1,-1) -- (-4, -1, 8) -- (5,9,7)-- (5,9,-3) -- cycle;

        \coordinate (O) at (0,0,0);

        \coordinate[label=above right:{$\im(\alpha_1)$}] (A) at (0,6,0);
        \coordinate[label=above left:{$\P_{\alpha_{1}\beta}$}] (B) at (0,2,-0.75);
        \coordinate[label=above right:{$\beta$}] (C) at (2, 2, 4);

        \draw[->] (O) -- (6,0,0);
        \draw[->] (O) -- (0,6,0);
        \draw[->] (O) -- (0,0,6);
    
        \draw[->,draw=orange,style=thick] (O) -- (2, 2, 4);

        \draw[dashed,->,draw=orange,style=thick] (O) -- (0, 2, 0);
        \draw[-,draw=blue,style=ultra thick,opacity=0.4] (O) -- (0, 7, 0);
        \draw[dotted,draw=black,style=thick, opacity=0.4] (2, 2, 4) -- (0, 2, 0);

        \coordinate (A) at (0, 2.25, 0);
        \coordinate (B) at (0.25, 2.25, 0.25);
        \coordinate (C) at (0.25, 2, 0.25);
        \draw[draw=black] (A) -- (B) -- (C);
        \end{tikzpicture}
    }
  \end{minipage}
\hspace{-2cm}
  \begin{minipage}[b]{0.4\linewidth}
    \centering
    \resizebox{\linewidth}{!}{
        \begin{tikzpicture}[>={Stealth[inset=0pt,length=8pt,angle'=35,round]},tdplot_main_coords,scale=.75]
        \coordinate (O) at (0,0,0);
        \draw[draw=white,fill=blue,fill opacity=0.0] (-4,-1,-1) -- (-4, -1, 8) -- (5,9,7)-- (5,9,-3) -- cycle;

        \coordinate[label=above right:{$\beta$}] (C) at (2, 2, 4);
    
        \coordinate[label=above right:{$\im(\alpha_2)$}] (D) at (7,0,0);
        \coordinate[label=above left:{$\P_{\alpha_{2}\beta}$}] (E) at (2,0,0);

        \draw[->] (O) -- (6,0,0);
        \draw[->] (O) -- (0,6,0);
        \draw[->] (O) -- (0,0,6);
    
        \draw[->,draw=orange,style=thick] (O) -- (2, 2, 4);

        \draw[dashed, ->,draw=orange,style=thick] (O) -- (2, 0, 0);
        \draw[-,draw=violet,style=ultra thick,opacity=0.4] (O) -- (7, 0, 0);
        \draw[dotted,draw=black,style=thick, opacity=0.4] (2, 2, 4) -- (2, 0, 0);

        \draw[->,draw=orange,style=thick] (O) -- (2, 2, 4);

        \coordinate (A) at (1.75,0 , 0);
        \coordinate (B) at (1.75, 0.25, 0.25);
        \coordinate (C) at (2.0, 0.25, 0.25);
        \draw[draw=black] (A) -- (B) -- (C);
            
        \end{tikzpicture}
    }
  \end{minipage}
   \hspace{-2cm}
    \begin{minipage}[b]{0.4\linewidth}
    \centering
    \resizebox{\linewidth}{!}{
        \begin{tikzpicture}[>={Stealth[inset=0pt,length=8pt,angle'=35,round]},tdplot_main_coords,scale=.75]
        \coordinate (O) at (0,0,0);
       \draw[draw=white,fill=blue,fill opacity=0.0] (-4,-1,-1) -- (-4, -1, 8) -- (5,9,7)-- (5,9,-3) -- cycle;
        
        \draw[draw=white,fill=blue,fill opacity=0.6,pattern=horizontal lines light blue] (O) -- (7,0,0) -- (7,7,0)-- (0,7,0) -- cycle;

        \coordinate[label=above right:{$\im(\alpha_{1/2})$}] (A) at (0,6,2);
        \coordinate[label=above left:{$\P_{\alpha_{1/2}\beta}$}] (B) at (2,2,0);
        \coordinate[label=above right:{$\beta$}] (C) at (2, 2, 4);

        \draw[->] (O) -- (6,0,0);
        \draw[->] (O) -- (0,6,0);
        \draw[->] (O) -- (0,0,6);
    
        \draw[->,draw=orange,style=thick] (O) -- (2, 2, 4);

        \draw[dotted,draw=black,style=thick, opacity=0.6] (2, 2, 4) -- (2, 2, 0);

        \draw[->,draw=orange,style=thick] (O) -- (2, 2, 4);
        \draw[dashed,->,draw=orange,style=thick,opacity=0.5] (O) -- (2, 2, 0);

        \coordinate (A) at (1.75, 1.75, 0);
        \coordinate (B) at (1.75, 1.75, 0.25);
        \coordinate (C) at (2, 2, 0.25);
        \draw[draw=black] (A) -- (B) -- (C);
        \end{tikzpicture}
    }
  \end{minipage}
  \hfill
  \vspace{-0.5cm}
  \caption{
   The illustration shows the possible experimental settings, including their estimates in a setting with $\dimZ=2, \dimX=3$: (1) The experiments with one of the two instruments respectively, denoted by $\alpha_1$ and $\alpha_2$ and (2) - if feasible - with both instruments at once, i.e. $\alpha_{1/2}$.
   In each figure, the true causal effect $\beta$ (orange) is projected onto the corresponding instrumented spaces, i.e. $\im(\alpha_{1}), \im(\alpha_{2})$ and $\im(\alpha_{1/2})$ (shaded blue) by $\P_{\alpha_{1}} \beta, \P_{\alpha_{2}} \beta$ and $\P_{\alpha_{1/2}} \beta$ (dashed orange). }
\label{fig:graphical_illustration1}
\end{figure*}

\section{Background and Problem Setting}

We aim to estimate the causal effect of treatments $X \in \bR^\dimX$ on a scalar outcome $Y \in \bR$.
There may be unobserved confounding between $X$ and $Y$, but we assume access to valid instruments $Z \in \bR^{\dimZ}$.
We focus on the linear setting
\begin{equation}\label{eq:setting}
  X = Z \alpha + \epsilon_X \:, \qquad
  Y = X \beta + \epsilon_Y\:,
\end{equation}
where $\alpha \in \bR^{\dimZ \times \dimX}$ and $\beta \in \bR^{\dimX}$ represent the linear structural functions. We interpret $\beta$ flexibly as a row or column vector as needed.
Because of unobserved confounding, the noise variables $\epsilon_X, \epsilon_Y$ are typically not independent, but we assume $\bE[\epsilon_X] = \bE[\epsilon_Y] = 0$. 
The standard instrumental variable assumptions (A1)-(A3) become $\alpha \ne 0$, $Z \indep \{\epsilon_X, \epsilon_Y\}$, and $Z \indep Y \given \{X, \epsilon_Y\}$.

Therefore, estimating the causal effect simply corresponds to estimating $\beta$.
The OLS estimator for $\beta$ will be biased, but if $\dimX \leq \dimZ$ and a rank condition on the covariance of $Z$ and $X$ is satisfied, $\beta$ is point-identified and can be estimated consistently, for example via the standard 2SLS estimator \citep{angrist2008mostly}.
When collecting $n$ i.i.d.\ observations in matrices $\X \in \bR^{n \times \dimX}$, $\Z \in \bR^{n \times \dimZ}$, and $\y \in \bR^n$, the 2SLS estimator for $\beta$ is given by
\begin{equation}\label{eq:2sls_estimate}
   \betaiv = (\X^T \P_{\Z} \X)^{-1} \X^T \P_{\Z} \y \:,
\end{equation}
where $\P_{\Z} = \Z (\Z^T \Z)^{-1} \Z^T$ is a projection matrix.
The 2SLS estimator can be viewed as (a) regressing $X$ on $Z$, and (b) regressing $Y$ on the predicted $X$ values from the first-stage regression.
Intuitively, the influence of $\epsilon_X$ on $X$ is ``regressed out'' in the first stage, leaving only the direct causal effect of $X$ on $Y$ in the second stage.

Following traditional application settings, most proposed IV methods assume the \emph{just-identified} setting $\dimX = \dimZ$ (as well as the rank condition on $\mathrm{cov}(X, Z)$, which translates into a full rank condition on $\alpha$).
Typically, these methods can also accommodate the \emph{overidentified} (or \emph{overspecified}) setting $\dimZ > \dimX$.
Having many instruments available in the overspecified setting can also be exploited for consistent estimators under relaxed assumptions such as correlated or weak instruments \citep{hausman2005validIV,kang2016instrumental}.
In contrast, the \emph{underidentified} (or \emph{underspecified}) case $\dimZ < \dimX$, where there are fewer instruments than treatments, has received little attention in the literature.
In this case, since $\rk(\P_{\Z}) \le \dimZ$, we have that $\rk(\X^T \P_{\Z} \X) \le \dimZ < \dimX$ and consequently cannot take the inverse in \cref{eq:2sls_estimate}.
More generally, in this situation the causal effect $\beta$ is not fully identified.
However, the available instruments still constrain the possible values of $\beta$.
We will later estimate and exploit those confines.

In our setting, we assume access to a fixed number of $\ninstruments \in \bN$ instruments (e.g., available antibiotics, or drugs) and---in high-dimensional treatment settings---typically have $\ninstruments < \dimX$.
Because we will also consider subsets of instruments, we generically denote by $\dimZ \in [\ninstruments] := \{1, \ldots, \ninstruments\}$ the number of IVs in a given estimation (or experiment).
Further, we assume experimental access to the instruments in that we can run experiments in which a chosen set of instruments is randomized (e.g., in mouse studies or on cell cultures).
Since the specific distribution of $Z$ does not affect identifiability, we assume without loss of generality that the components of $Z$ all independently follow a Rademacher distribution.
This can be interpreted as whether a drug or antibiotic is applied.\footnote{Other illustrative choices are $\tfrac{1}{2}$-Bernoulli or the uniform distribution on a finite interval representing (standardized) dosage levels. All our results immediately transfer.}
With this choice, $\alpha$ fully characterizes the effect of $Z$ on $X$ and we consequently also call the rows of $\alpha \in \bR^{\ninstruments \times \dimX}$ ``the instruments''.\footnote{Generically, $\alpha \in \bR^{\dimZ \times \dimX}$ with $\dimZ \le \ninstruments$ represents some choice of $\dimZ$ instruments in a given estimation/experiment.}
Since the components of $Z$ are jointly independent, the full rank condition on $\mathrm{cov}(X, Z)$ translates to $\alpha$ having full rank.

To maximally constrain the effect estimate, one would randomize all $\ninstruments$ available instruments simultaneously.
Due to (B3), this is typically infeasible in practice.
We model this via a cost function, which can also incorporate hard constraints such as limiting the number of instruments per experiment to at most $\nperround < \ninstruments$.
For (B2) we limit the total number of possible experiments to $\nexperiments \in \bN$.

We proceed as follows.
In \cref{sec:underidentified}, we first construct a consistent, asymptotically normal estimator for the orthogonal projection of $\beta$ onto the image of $\alpha$ viewed as a linear map $\alpha: \bR^\dimZ \to \bR^\dimX$.
We then establish a method to combine multiple such estimates obtained from (different) subsets of IVs to obtain a consistent estimate for the orthogonal projection of $\beta$ on the linear subspace spanned by all instruments combined.
Furthermore, we introduce a method to determine which components of $\beta$ have been successfully identified and a condition that guarantees full identification of $\beta$ under mild assumptions.
In \cref{sec:sequential}, we then leverage all these findings to develop a procedure that proposes subsets of instruments for sequential experimentation to maximally identify $\beta$ under the given constraints.

\Cref{fig:graphical_illustration1} and \cref{fig:graphical_illustration2} illustrate the estimation resp. the combination step and the role of linearity in this context visually.
We suggest to consult these illustrations for intuition when reading the formal statements in the following section.

\begin{figure}
\begin{minipage}[b]{0.8\linewidth}
    \centering
    \resizebox{\linewidth}{!}{
        \begin{tikzpicture}[>={Stealth[inset=0pt,length=8pt,angle'=35,round]},tdplot_main_coords,scale=0.75]
        \coordinate (O) at (0,0,0);
        \draw[draw=white,fill=blue,fill opacity=0.0] (-4,-1,-1) -- (-4, -1, 8) -- (5,9,7)-- (5,9,-3) -- cycle;
        \draw[draw=white,fill=blue,fill opacity=0.6,pattern=horizontal lines light blue] (O) -- (7,0,0) -- (7,7,0)-- (0,7,0) -- cycle;

        \coordinate[label=above right:{$\im(\alpha_{1/2})$}] (A) at (0,6,2);
        \coordinate[label=above left:{$\P_{\alpha_{1/2}\beta}$}] (B) at (2,2,0);
        \coordinate[label=above right:{$\beta$}] (C) at (2, 2, 4);

        \coordinate[label=above left:{$\P_{\alpha_{1}\beta}$}] (B) at (0,2,-0.75);
        \draw[dashed,->,draw=orange,style=thick,opacity=0.4] (O) -- (0, 2, 0);
        \coordinate[label=above left:{$\P_{\alpha_{2}\beta}$}] (E) at (2,0,0);
        \draw[dashed, ->,draw=orange,style=thick,opacity=0.4] (O) -- (2, 0, 0);

        \draw[->] (O) -- (6,0,0);
        \draw[->] (O) -- (0,6,0);
        \draw[->] (O) -- (0,0,6);
    
        \draw[->,draw=orange,style=thick] (O) -- (2, 2, 4);

        \draw[dotted,draw=black,style=thick, opacity=0.6] (2, 2, 4) -- (2, 2, 0);

        \draw[->,draw=orange,style=thick] (O) -- (2, 2, 4);
        \draw[dashed,->,draw=orange,style=thick,opacity=0.5] (O) -- (2, 2, 0);

        \coordinate (A) at (1.75, 1.75, 0);
        \coordinate (B) at (1.75, 1.75, 0.25);
        \coordinate (C) at (2, 2, 0.25);
        \draw[draw=black] (A) -- (B) -- (C);
        \end{tikzpicture}
    }  
        
\end{minipage}
  \begin{minipage}[b]{0.8\linewidth}
    \centering
    \resizebox{\linewidth}{!}{
        \begin{tikzpicture}[>={Stealth[inset=0pt,length=8pt,angle'=35,round]},tdplot_main_coords,scale=0.75]
        \draw[draw=white,fill=blue,fill opacity=0.0] (-4,-1,-1) -- (-4, -1, 8) -- (5,9,7)-- (5,9,-3) -- cycle;
        
        \coordinate (O) at (0,0,0);
        \coordinate[label=above right:{$\im(\alpha_1)$}] (A) at (0,6,0);
        \coordinate[label=above right:{$\beta$}] (C) at (2, 2, 4);
        
        \coordinate[label=above right:{$\im(\alpha_2)$}] (D) at (7,0,0);
        \coordinate[label=above left:{$\P_{\alpha_{[1,2]}\beta}$}] (F) at (1.5,5,0);

        \draw[->] (O) -- (6,0,0);
        \draw[->] (O) -- (0,6,0);
        \draw[->] (O) -- (0,0,6);
    
        \draw[->,draw=orange,style=thick] (O) -- (2, 2, 4);

        \draw[dashed,->,draw=orange,style=thick] (O) -- (0, 2, 0);
        \draw[-,draw=blue,style=ultra thick,opacity=0.4] (O) -- (0, 7, 0);
        \draw[draw=white,fill=blue,fill opacity=0.15] (0,2,0) -- (3, 2, 0) -- (3,2,5)-- (0,2,5) -- cycle;

        \draw[dashed,->,draw=orange,style=thick] (O) -- (2, 0, 0);
        \draw[-,draw=violet,style=ultra thick,opacity=0.4] (O) -- (7, 0, 0);
        \draw[draw=white,fill=violet,fill opacity=0.15] (2,4,0) -- (2, 4, 5) -- (2,0,5)-- (2,0,0) -- cycle;

         \draw[draw=white,fill=blue,fill opacity=0.4] (2,2,0) -- (2, 2, 5) -- (0,2,5)-- (2,0.8,0) -- cycle;

        \draw[dashed,draw=blue,style=thick,opacity=0.7] (2,2,6) -- (2, 2, -1);
        \draw[->,draw=orange,style=thick] (O) -- (2, 2, 4);
        \draw[dotted,draw=black,style=thick, opacity=0.4] (2, 2, 0) -- (0, 2, 0);
        \draw[dotted,draw=black,style=thick, opacity=0.4] (2, 2, 0) -- (2, 0, 0);

        \draw[dashed,->,draw=orange,style=thick] (O) -- (2, 2, 0);

        \coordinate (A) at (1.75, 0, 0);
        \coordinate (B) at (1.75, 0.25, 0);
        \coordinate (C) at (2, 0.25, 0);
        \draw[draw=black] (A) -- (B) -- (C);

        \coordinate (A) at (0, 1.75, 0);
        \coordinate (B) at (0.25, 1.75, 0);
        \coordinate (C) at (0.25, 2, 0);
        \draw[draw=black] (A) -- (B) -- (C);
        \end{tikzpicture}
    }
  \end{minipage}
  \vspace{-0.5cm}
  \caption{The illustration shows a setting with $\dimZ=2, \dimX=3$. It contrasts the single estimation steps with the combination step.
  Upper panel, single estimation: the true causal effect $\beta$ (orange) is projected onto the individual instrumented spaces by $\P_{\alpha_{1}} \beta, \P_{\alpha_{2}} \beta$ and $\P_{\alpha_{1/2}} \beta$.
  Lower panel, combination step: The intersection (dashed blue line) of the orthogonal complements of $\{ \im(\alpha_{1}), \im(\alpha_{2})\}$ represents all vectors that are compatible with both individual estimates. \Cref{eq:underiv_combinedestimate_optimizationprob} then selects the minimum norm point on that line.
  In the illustration it comes to show that the combination of individual estimators $\P_{\alpha_{[1,2]}}\beta$ recovers the effect of the single estimation with both instruments $\P_{\alpha_{1/2}}\beta$.}  \label{fig:graphical_illustration2}
\end{figure}
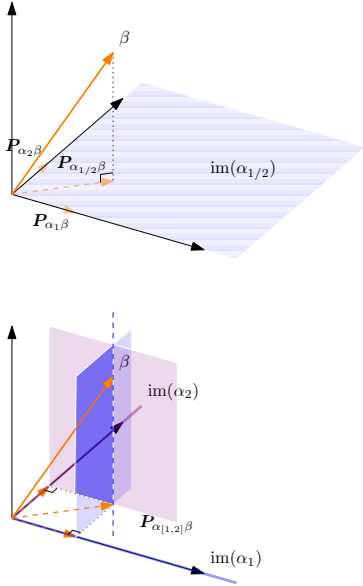

\section{Underidentified IV Estimates}
\label{sec:underidentified}

\subsection{Estimator for a Single Experiment}

Consider a set of instruments $\alpha \in \bR^{\dimZ \times \dimX}$ with $\dimZ < \dimX$ and $\rk(\alpha) = \dimZ$.
The predicted treatment values via OLS $\hat{\X} = \Z (\Z^T \Z)^{-1} \Z^T \X = \P_{\Z} \X \in \bR^{n \times \dimX}$ are confined to a proper linear subspace of $\bR^{\dimX}$, rendering $\X^T P_{\Z} \X$ singular and $\betaiv$ in \cref{eq:2sls_estimate} ill-defined in the underspecified setting.
With $\hat{\alpha}$ being the first-stage OLS estimate, we have $\hat{\X}_i = \P_{\hat{\alpha}} \X_i \in \im(\hat{\alpha}) := \{ \hat{\alpha} z \given z \in \bR^{\dimZ}\} \subset \bR^{\dimX}$ for all $i \in [n]$, where $\P_{\hat{\alpha}} = \hat{\alpha}^T (\hat{\alpha} \hat{\alpha}^T)^{-1} \hat{\alpha} \in \bR^{\dimX \times \dimX}$ denotes the orthogonal projection onto $\im(\hat{\alpha})$.
That is, all first-stage predictions lie in the image of $\hat{\alpha}$ and we can consider the second stage as a mapping $\bR^{\dimX} \supset \im(\hat{\alpha}) \to \bR$.
We call $\im(\hat{\alpha})$ the \emph{instrumented subspace} (in this experiment).
Intuitively, while the low-dimensional $Z$ cannot ``shake'' or ``instrument'' the entire treatment space to identify $\beta$, it still ``instruments'' a non-trivial subspace, inducing non-trivial constraints on $\beta$.
In particular, one may expect to recover ``the part of $\beta$ within the instrumented subspace''.
The following statement formalizes this intuition.

\begin{proposition}\label{prop:identifyprojection}
  Let $\alpha \in \bR^{\dimZ \times \dimX}$ have full rank and assume we have i.i.d.\ data $\Z, \X, \y$ from an experiment in which all $\dimZ$ instruments have been randomized.
  Then
  \begin{align}\label{eq:betaivsub}
      \betaivsub{} &:= (\X^T \P_{\Z} \X)^+ \X^T \P_{\Z} \y \overset{d}{\longrightarrow} \cN(\P_\alpha \beta, \Sigma)\\
      \text{with }\: 
      \Sigma &:= \tfrac{1}{n} \alpha^+ \Sigma_Z^{-1} (\alpha^T)^+ \Var[\epsilon_Y] {\color{gray} {}= \tfrac{1}{n} (\alpha^T \alpha)^+ \Var[\epsilon_Y]} \:,\nonumber
\end{align}
  where $(\cdot)^{+}$ denotes the Moore-Penrose pseudoinverse, $\Sigma_Z$ is the covariance matrix of $Z$, and the simple form of $\Sigma$ (in gray) applies when $Z$ are i.i.d.\ Rademacher variables.
\end{proposition}

\begin{proof}
  For $\dimZ \ge \dimX$, the pseudoinverse in \cref{eq:betaivsub} can be replaced with a regular matrix inverse and the result follows from the asymptotic normality of $\betaiv$ in \cref{eq:2sls_estimate} \citep[Sec.\ 4.2.1]{angrist2008mostly}.
  For the underspecified case $\dimZ < \dimX$,
  we start with the singular value decomposition (SVD) $\hat{\X} = \U \D \V^T$, with which $\hat{\X}^T \hat{\X} = \V \D^T \D \V^T$.
  Assuming that the singular values in the rectangular diagonal matrix $\D \in \bR^{n \times \dimX}$ are sorted in non-ascending order, only the first $\dimZ$ entries are nonzero, because $\rk(\hat{\X}) = \dimZ$.
  Accordingly, the first $\dimZ$ rows of $\V$ form an orthonormal basis of $\im(\alpha)$.
  We write $\V_{\hat{\alpha}} \in \bR^{\dimX \times \dimZ}$ for the first $\dimZ$ columns of $\V$, $\D_{\hat{\alpha}} \in \bR^{\dimZ \times \dimZ}$ for the upper left $\dimZ \times \dimZ$ block of $\D$, and $\U_{\hat{\alpha}} \in \bR^{n\times \dimZ}$ for the first $\dimZ$ columns of $\U$.
  With $\hat{\X} = \P_{\Z} \X$ we compute
  \begin{align*}
      \betaivsub{} 
      &= (\X^T \P_{\Z}^T \P_{\Z} \X)^+ \X^T \P_{\Z}^T \P_{\Z} (\X \beta + \epsilon_Y)  \\
      &= (\hat{\X}^T \hat{\X})^+ (\hat{\X}^T\hat{\X} \beta + \hat{\X}^T \epsilon_Y) \\
      &= \V (\D^T \D)^+ \D^T\D \V^T \beta + \V (\D^T)^+ \U^T \epsilon_Y \\
      &=  \V_{\hat{\alpha}} \V_{\hat{\alpha}}^T \beta + \V_{\hat{\alpha}} \D_{\hat{\alpha}}^{-1} \U_{\hat{\alpha}}^T \epsilon_Y\:,
\end{align*}
  where we have used 
  \begin{equation*}
    (\D^T \D)^+ \D^T\D = \begin{pmatrix}
         \bm{I}_{\dimZ \times \dimZ}& \bm{0}_{\dimZ \times (\dimX - \dimZ)} \\
         \bm{0}_{(\dimX - \dimZ) \times \dimZ} & \bm{0}_{(\dimX - \dimZ) \times (\dimX - \dimZ)}
    \end{pmatrix}\:.
  \end{equation*}
  From the asymptotic normality of the first-stage OLS estimate $\hat{\alpha}$, it follows that $\V_{\hat{\alpha}}$ and $\D_{\hat{\alpha}}$ are $\sqrt{n}$-consistent estimators for $\V_{\alpha}$ and $\D_{\alpha}$ respectively \citep{bura2008distribution}.
  Therefore, $\P_{\hat{\alpha}} = \V_{\hat{\alpha}} \V_{\hat{\alpha}}^T$ converges in probability to $\P_{\alpha}$ in the large sample limit.
  With $\E[\epsilon_Y] = 0$ we thus have
  \begin{equation*}
      \E[\betaivsub{} ] = \E[\V_{\hat{\alpha}} \V_{\hat{\alpha}}^T] \beta \overset{p}{\longrightarrow} \P_\alpha \beta \:.
  \end{equation*}
  Similarly, with $\hat{\X} = \Z \hat{\alpha}$ we compute the covariance
  \begin{align*}
      \Cov[\betaivsub{}] 
      &= \Cov[\V_{\hat{\alpha}} \D_{\hat{\alpha}}^{-1} \U_{\hat{\alpha}}^T \epsilon_Y]
      = \hat{\X}^+ (\hat{\X}^T)^+ \Var[\epsilon_Y] \\
      &= \hat{\alpha}^+ (\Z^T \Z)^{-1} (\hat{\alpha}^T)^+ \Var[\epsilon_Y]\:.
  \end{align*}
  Since the matrix inverse is continuous (on invertible matrices), the continuous mapping theorem yields that $\hat{\Sigma}_{Z}^{-1}$ consistently estimates $\Sigma_Z^{-1}$.
  Again using the asymptotic normality of $\hat{\alpha}$, we thus have (for centered) $\Cov[\betaivsub{}] \overset{p}{\longrightarrow} \tfrac{1}{n} \alpha^+ \Sigma_Z^{-1} (\alpha^T)^+ \Var[\epsilon_Y]$.
  For i.i.d.\ Rademacher instruments with covariance matrix $\Sigma_Z = \bm{I}_{\dimZ \times \dimZ}$ we have $\Cov[\betaivsub{}] \overset{p}{\longrightarrow} \tfrac{1}{n} (\alpha^T \alpha)^+$.
\end{proof}

\Cref{prop:identifyprojection} substantiates the first intuition that in the underidentified setting we can still consistently estimate the orthogonal projection of $\beta$ onto the instrumented subspace.

We refer to \cref{fig:graphical_illustration1} for an illustration of the estimation in a linear setting with $\dimZ=2, \dimX=3, \nexperiments=2$.
For visualization purposes, we choose axis-aligned $\alpha_{1} = (1, 0, 0)$ and $\alpha_{2} = (0, 1, 0)$.
For both instruments, we estimate the projection of $\beta$ onto the respective instrumented spaces ($P_{\alpha_{1}}\beta$ and $P_{\alpha_{2}}\beta$). The illustration on the right hand side also includes the estimation which uses both instruments ($\alpha_{1/2}$) at once. In the following, this is the effect we want to recover solely based on the individual estimates.

\subsection{Combined Estimator}

In sequential experiments, we select pairwise disjoint instruments $\alpha_1, \alpha_2, \ldots, \alpha_{\nexperiments}$, where each $\alpha_i$ is a subset of the $\ninstruments$ available ones.
Each corresponding individual estimator $\betaivsub{1}, \betaivsub{2}, \ldots, \betaivsub{\nexperiments}$ estimates an orthogonal projection of $\beta$ onto the linear subspace $\im(\hat{\alpha}_i)$.
Our goal is to reconstruct the estimate we would have obtained, had we used all instruments in the sequential experiments at once in a single one.
Denoting the union (or concatenation) of all available instruments by a single matrix $\alpha_{[\nexperiments]}$, we aim to reconstruct $\P_{\alpha_{[\nexperiments]}} \beta$ from the individual $\betaivsub{i}$.\footnote{We use similar notation $\hat{\alpha}_{[t]}$ for the concatenation of the instruments (as matrices) used up until and including round $t$.}
In other words, we are looking for the least-square vector compatible with all projections
\begin{equation}  \label{eq:underiv_combinedestimate_optimizationprob}
    \min_{\gamma \in \bR^\dimX}\|\gamma \|_2^2 \: \text{ s.t.} \:\:  \betaivsub{i} = \V_{\hat{\alpha}_i} \V_{\hat{\alpha}_i}^T \gamma \: \text{ for all } i \in [\nexperiments]\:.
\end{equation}

\begin{proposition}\label{prop:combinedestimate}
  Let $\alpha \in \bR^{\ninstruments \times \dimX}$ have full rank and let $(\Z_1, \X_1, \y_1), \dots, (\Z_{\nexperiments}, \X_{\nexperiments}, \y_{\nexperiments})$ be i.i.d.\ datasets from $\nexperiments$ experiments with disjoint subsets $\alpha_1, \ldots, \alpha_{\nexperiments}$ of randomized instruments. 
  Then the solution of \cref{eq:underiv_combinedestimate_optimizationprob} is a consistent estimator for $\P_{\alpha_{[\nexperiments]}} \beta$.
\end{proposition}
\begin{proof}
  Let $\bm{A} := \V_{\hat{\alpha}_{[\nexperiments]}} \V_{\hat{\alpha}_{[\nexperiments]}}^T =  (\V_{\hat{\alpha}_1} \V_{\hat{\alpha}_1}^T | \ldots | \V_{\hat{\alpha}_{\nexperiments}} \V_{\hat{\alpha}_{\nexperiments}}^T) \in \bR^{(\nexperiments \dimX) \times \dimX}$ and $\bm{b} := (\betaivsub{1}, \ldots, \betaivsub{\nexperiments}) \in \bR^{\nexperiments \dimX}$, where ``$|$'' denotes concatenation.
  Then the solution to \cref{eq:underiv_combinedestimate_optimizationprob} is simply $\bm{A}^+ \bm{b}$,  i.e., the least-squares solution to $\bm{A} \gamma = \bm{b}$.
  Hence the optimal $\gamma$ is the orthogonal projection of $\bm{b}$ onto $\im(\bm{A})$, which by construction is just $\im(\hat{\alpha}_{[\nexperiments]})$.
  By the consistency of $\hat{\alpha}_i$ (and thus the consistency of $\bm{A}$ and $\bm{b}$ for their respective expressions without hats as in the proof of \cref{prop:identifyprojection}), asymptotically $\bm{A}^+ \bm{b} \overset{p}{\longrightarrow} \P_{\alpha_{[\nexperiments]}} \beta$.
\end{proof}

Note that the instrument sets need not be disjoint for the proof of \cref{prop:combinedestimate}.
However, since we do not gain further information by including the same instrument in two distinct experiments, which would instrument the same subspace twice, it is reasonable to keep instrument sets distinct across experiments for efficiency.
The least-squares problem with linear equality constraints in \cref{eq:underiv_combinedestimate_optimizationprob} can be solved efficiently (in high dimensions and for many constraints) by simply solving a linear system \citep[Sec.\ 16.1]{nocedal2006numerical}.
In practice, we keep a running estimate $\betaivsub{[t]}$, which is the combined estimate for all experiments up to and including round $t \in [\nexperiments]$.
Being able to combine estimates from individual sets of instruments ``as if we had run a single experiment using the union of instruments at once'' forms the basis for leveraging sequential experiments with different sets of randomized instruments to optimally constrain the overall causal effect in \cref{sec:sequential}.

\subsection{Full Identification and Identified Components}

\xhdr{Full identification}
As an interesting special case of \cref{prop:identifyprojection}, we note that regardless of $\dimX$, a single instrument $\dimZ = 1$ in principle suffices to identify $\beta \in \bR^{\dimX}$, namely when $\alpha \in \bR^{\dimX \times 1}$ is parallel to $\beta$ (viewed as vectors in $\bR^\dimX$), which implies $\P_{\alpha} \beta = \beta$.
Of course, this is unlikely to ``happen accidentally'' in practice.
However, it highlights that if we had full control over the instruments, we could fully identify $\beta$ by crafting $\alpha \in \bR^{\dimX \times 1}$ such as to maximize the $2$-norm of the resulting estimated $\betaivsub{}$.
This can be seen from the fact that $\|\P_{\alpha} \beta \|_2 \le \|\beta\|_2$ for any $\alpha$ with equality implying $\beta \in \im(\alpha)$.
Therefore, when running sequential experiments, the $2$-norm of the combined estimate is a proxy for whether we are still gaining relevant information.
We summarize these results in the following Corollary.
\begin{corollary}\label{cor:recoverbeta}
 In the setting of \cref{prop:identifyprojection}, the following are equivalent: (i) $\beta \in \im(\alpha)$; (ii) $\| \betaivsub{} \|_2 \longrightarrow \| \beta\|_2$ as $n \to \infty$; (iii) $\betaivsub{}$ consistently estimates $\beta$ (not just $\P_{\alpha} \beta$).
\end{corollary}

In our underspecified setting, it is generally impossible to determine whether any (and therefore all) of the statements in \cref{cor:recoverbeta} hold true.
However, recently \citet{janzing2018detecting} have shown that under mild assumptions on the confounding model (i.e., on the joint distribution of $\epsilon_X, \epsilon_Y$), one can estimate what they call the \emph{confounding strength}.
Their estimator has been further refined by \citet{rendsburg2022consistent}.
The confounding strength can then be leveraged to estimate $\| \beta \|_2$ itself from purely observational data \citep{janzing2019causal}.
This means that we can consistently estimate $\| \beta \|_2$ from data $(\X, \y)$, where no instruments have been randomized in our setting \citep[step 5 in Alg.\ ConCorr]{janzing2019causal}.\footnote{Since we assume experimental access, we can also collect purely observational $(X, Y)$ data about the system of interest.}
Let us denote this estimator by $\betahat{}$.
We outline the assumptions on the confounding model, which are satisfied in our empirical evaluation, in \cref{sec:experiments}.

Following \cref{cor:recoverbeta} we can use the difference of the $2$-norm of our running estimate $\|\betaivsub{[t]}\|_2$ and $\betahat{}$ as a heuristic for whether we have already fully identified $\beta$, i.e., whether the overall instrumented subspace contains $\beta$.\footnote{Recall that the $2$-norm of an orthogonal projection of $\beta$ is always smaller or equal to the $2$-norm of $\beta$.
This provides additional motivation for why we choose to combine estimates according to \cref{eq:underiv_combinedestimate_optimizationprob} using a minimum $2$-norm objective.}
In future work, with the asymptotic normality of $\betaivsub{}$ in \cref{prop:identifyprojection} and thus an explicit expression for the asymptotic distribution of $\|\betaivsub{}\|_2^2$ \citep{moschopoulos1985distribution} as well as the asymptotic behavior of $\betahat{}$ in \citet{rendsburg2022consistent}, one could derive (asymptotic) confidence intervals for whether $||\betaivsub{[t]}|| = \betahat{}$.
In practice, we propose
\begin{equation}\label{eq:stopping}
    \left|\betahat{} - \|\betaivsub{[t]}\|_2 \right| < \epsilon
\end{equation}
for a fixed tolerance $\epsilon > 0$ as a stopping criterion of our algorithm.
When this condition is reached, we can conclude (for sufficiently large $n$) that $\beta$ is near fully identified, even though we may have used fewer than $\dimX$ instruments.

\xhdr{Identified components}
When the above stopping criterion is not reached within our $\nexperiments$ rounds of experimentation, learning an orthogonal projection of $\beta$ may still be informative in itself.
However, in many situations one is interested in precise values of individual components $\beta_i$ for $i \in [\dimX]$.
Notably, when $\beta \notin \im(\alpha)$, the individual components of $\P_{\alpha}\beta$ are not easily interpretable.
For example, neither holds $(\P_{\alpha} \beta)_i = 0 \Rightarrow \beta_i = 0$ nor the other way round. 
Therefore, we now devise a method to determine when we can trust any given component of our running estimate $\betaivsub{[t]}$.
\begin{corollary}\label{cor:components}
  Let $(e_i)_{i\in [\dimX]}$ be the standard basis of $\bR^{\dimX}$.
  In the setting of \cref{prop:combinedestimate}, when $\V_{{\alpha}_{[t]}} \V_{{\alpha}_{[t]}}^T e_i = e_i$, then $(\betaivsub{[t]})_i$ consistently estimates $\beta_i$.
\end{corollary}
\begin{proof}
   This follows from \cref{prop:combinedestimate} and the linearity of projections when considering $\beta = \sum_{i=1}^{\dimX} \beta_i e_i$.
\end{proof}
We note that we already compute $\V_{\hat{\alpha}_t}$ at each round as part of $\betaivsub{t}$.
Since $\V_{\hat{\alpha}_{[t]}} \V_{\hat{\alpha}_{[t]}}^T \overset{p}{\longrightarrow} \P_{\alpha_{[t]}}$, the check for identified components in \cref{cor:components} can therefore be performed efficiently in practice by checking for (approximate) equality of $\V_{\hat{\alpha}_{[t]}} \V_{\hat{\alpha}_{[t]}}^T e_i \approx e_i$.
In our empirical evaluation, we use the absolute value of the cosine similarity, denoted by \texttt{cdist}, between the two vectors as a continuous measure for whether $\beta_i$ has been identified.
As an aggregate metric, we report the percentage of identified components
\begin{equation}\label{eq:identifiedcomps}
  \tfrac{1}{\dimX} \left|\{i \in [\dimX] \given \texttt{cdist}(V_{\alpha[t]} V_{\alpha[t]}^T e_i, e_i) < \delta\}\right|
\end{equation}
for some fixed tolerance $\delta > 0$.

Finally, we can estimate an upper bound on the absolute error in each non-identified component.
Let $\beta = \P_{\alpha}\beta + \nu$ be the orthogonal decomposition of $\beta$ into the instrumented subspace $\im(\alpha)$ and its orthogonal complement.
Since $\|\beta\|_2^2 = \|\P_{\alpha}\beta\|_2^2 + \|\nu\|_2^2$ and $|\nu_i| \le \|\nu\|_2$ for all components $i \in [\dimX]$, we have
\begin{equation}
    |\nu_i | \le \sqrt{\|\beta \|_2^2 - \| \P_{\alpha}\beta \|_2^2} \: \text{ for all } i \in [\dimX]\:.
\end{equation}
With our consistent estimates $\betahat{}$ and $\betaivsub{}$, we can thus upper bound all remaining unidentified components.

We return now to \cref{fig:graphical_illustration2} which illustrates the estimation (upper panel) and combination steps (lower panel) in our linear setting for $\dimZ=2, \dimX=3, \nexperiments=2$.
For both instruments, we estimate the projection of $\beta$ onto the respective instrumented spaces ($P_{\alpha_{1}}\beta$ and $P_{\alpha_{2}}\beta$), including the effect $P_{\alpha_{1/2}}\beta$ which is the one that should be recovered.
In the lower panel, the two planes are the orthogonal complements of the instrumented spaces and their intersection corresponds to all vectors that are compatible, i.e., would be projected onto $\im(\alpha_{1})$ and $\im(\alpha_{2})$ respectively.
This corresponds to the constraints in \cref{eq:underiv_combinedestimate_optimizationprob}.
Among those, we then select the vector with the smallest norm as our combined estimate. 
The figure both illustrates the necessity of linearity for the combination of estimates and the increasing norm of the combined estimate converging to $\|\beta\|$.

\section{Sequential Selection of Instruments}
\label{sec:sequential}

At each step $t \in [\nexperiments]$, we seek to select the most informative subset of instruments $\alpha_t$ out of the pool of $\ninstruments$ available choices.
After each round we combine the newly obtained estimate $\betaivsub{t}$ with all previous ones to obtain $\betaivsub{[t]}$.
Regarding consideration (B3), we assume a cost function $\cost: [\ninstruments] \to \bR_{\ge 0}$, where $\cost(\dimZ)$ is the cost of running a single experiment with $\dimZ$ randomized instruments.
For instance, this may be the actual monetary and logistic cost of randomly administering the selected drugs to a collection of $n$ cell cultures ($\Z$), sequencing the cells ($\X$), and measuring the outcome of interest ($\y$) for each culture.
However, the cost may also incorporate a hard limit $\nperround$ on the number of instruments that can sensibly be combined in a single experiment (e.g., without killing the organism), by setting $c(d) = \infty$ for $d > \nperround$.

In light of \cref{eq:stopping}, the ultimate goal is to select instruments that maximize $\|\betaivsub{[\nexperiments]}\|_2$.
However, in round $t$ we cannot anticipate by how much $\|\betaivsub{[t-1]}\|_2$ is going to increase for a candidate set of instruments $\alpha_t$ without actually performing the experiment.
Therefore, we must rely on another signal to sequentially select subsets of instruments.
Without any information about the available instruments, we cannot do better than selecting instruments (uniformly) at random at each step.
While in practice one may not have precise information about the actual effects of individual instruments, information about the similarity of different antibiotics or drugs is typically still available.
For example, certain antibiotics may have related active agents (high similarity) or certain drugs may target similar pathways (high similarity).
We assume that pairwise normalized similarities $\similarity_{i,j} \in [0,1]$ are provided for all available instruments $i,j \in [\ninstruments]$.
Here, $\similarity_{i,j} = \similarity_{j, i}$ and $\similarity_{i,i} = 1$.
Such similarities could also be derived from a set of features known about the instruments.
To optimally explore the treatment space, it is then natural to attempt to sequentially select highly dissimilar instruments.

Hence, to evaluate the expected gain of adding a new set of instruments $\cI \subset [\ninstruments]$ to the already used instruments $\cJ \subset [\ninstruments]$, we define the following gain function
\begin{align}
    \gain:
    \; & 2^{[\ninstruments]} \times 2^{[\ninstruments]} \to \bR_{\ge 0}\:,\\
    & (\cI, \cJ) \mapsto \frac{1}{|\cI|+|\cJ| - 1} \sum_{i \in \cI} \sum_{j \in \cI \cup \cJ} (1 - \similarity_{i,j})\:.\nonumber
\end{align}
This takes into account the similarities of the newly proposed instrument set within itself as well as with respect to the previously used ones.
For example, when all instruments are maximally dissimilar ($\similarity_{i,j} = \delta_{i,j}$), $\gain(\cI, \cJ) = |\cI|$ regardless of $\cJ$.
When all instruments are equal ($\similarity_{i,j} = 1$), $\gain(\cI, \cJ) = 0$ for all inputs.
Finally, we define the score of the set $\cI$ when the set $\cJ$ has already been used by 
\begin{equation}\label{eq:score}
  \score(\cI, \cJ) := \gain(\cI, \cJ) - \cost(|\cI|)\:.
\end{equation}
These considerations lead to the following setting.
At each round $t$, we select a subset $\cI_t \subset [\ninstruments] \setminus \cI_{[t-1]}$ of still unused instruments that maximize the score function given the already used instruments.
We run a randomized experiment with those instruments to collect data, estimate (a projection of) $\beta$ from this data (\cref{prop:identifyprojection}), and combine the estimate with the previous ones (\cref{prop:combinedestimate}).
By convention, we have $\cI_{[0]} = \cI_{\emptyset} = \emptyset$ and we overload terminology to call both $\alpha_t \in \bR^{|\cI_t| \times \dimX}$ and $\cI_t \subset [\ninstruments]$ ``the instruments selected at round $t$''.
We similarly use $\alpha_{[t]}$ and $\cI_{[t]}$ for the instruments selected up to (and including) round $t$.
When the stopping criterion (\cref{eq:stopping} is satisfied, we return our current estimate as an estimate of the full $\beta$.
Otherwise, we return the estimate after $T$ experiments together with the identified components (\cref{cor:components}).
We outline our sequential instrument selection (SIS) procedure in \cref{alg:cap}.

We remark that since $\gain(\cI, \cJ) \le |\cI|$ it makes sense to use a sublinear cost function (until a potential hard limit $\nperround$).
Intuitively, while it becomes more expensive to include multiple randomized variables in a single experiment, it is still cheaper than running an individual randomized experiment for each instrument separately.
The precise choice of the cost function is informed by the actual experimental setting and determines the trade-off between randomizing many instruments at once (to increase the dimensionality of the instrumented subspace) and the cost of doing so.

\begin{algorithm}
\caption{Sequential selection of instrument sets}\label{alg:cap}
\begin{algorithmic}[1]
\Require maximum rounds $\nexperiments$, pairwise similarities $\similarity \in [0,1]^{\ninstruments \times \ninstruments}$, cost function $\cost$, tolerance $\epsilon > 0$
\State collect observational data $\X, \y$
\State compute $\betahat$ from $\X, \y$ \Comment{\citet[ConCorr]{janzing2019causal}}
\State $\cC \gets \emptyset$ \Comment{set of identified components}
\For{$t \in [\nexperiments]$} \Comment{experimental rounds}
  \State $\cI_t \gets \argmax_{\cI \subset [\ninstruments] \setminus \cI_{[t-1]}} \score(\cI, \cI_{[t-1]})$
  \State collect $\Z_t, \X_t, \y_t$ \Comment{run experiment with $\cI_t$}
  \State $\betaivsub{t} \gets (\X_t^T \P_{\Z_t} \X_t)^{+} \X_t^T \P_{\Z_t} \y_t$ \Comment{\cref{prop:identifyprojection}}
  \State $\betaivsub{[t]} \gets \argmin_{\gamma \in \bR^{\dimX}} \|\gamma\|_2$ \Comment{\cref{prop:combinedestimate}}
  \Statex \hspace{2cm} s.t.\ \:$\betaivsub{\tau} = \V_{\hat{\alpha}_\tau}\V_{\hat{\alpha}_\tau}^T \gamma$ for all $\tau \in [t]$
  \If{$| \|\betaivsub{[t]} - \betahat{} | < \epsilon$} \Comment{fully identified, \eqref{eq:stopping}}
    \State $\cC \gets [\dimX]$
    \State \textbf{return} $\betaivsub{[t]}, \cC$
  \EndIf
  \State $\cC \gets \{i \in [\dimX] \given \V_{\hat{\alpha}_{[t]}}\V_{\hat{\alpha}_{[t]}}^T e_i \approx e_i\}$ \Comment{\cref{cor:components}}
\EndFor
\State \textbf{return} $\betaivsub{[\nexperiments]}, \cC$ \Comment{estimate, identified components}
\end{algorithmic}
\end{algorithm}

\section{Empirical Evaluation}\label{sec:experiments}

\xhdr{Setup}
Since a real-world evaluation of our approach would require access to sequential randomized experimentation in a complex setting, we are restricted to simulation studies.
We first illustrate the properties of our proposed (combined) causal effect estimators in the underspecified IV setting and then evaluate our full sequential instrument selection method.
We generate the parameters $\alpha$, $\beta$ randomly (\cref{app:experiment_details}) in order to avoid parameter selection bias.
Next, we fix a mixing matrix $\M \in \bR^{\dimX \times \dimX}$ with all entries sampled from independent standard Gaussians as well as a direction $v \in \bR^{\dimX}$ as a uniform sample from the sphere $\mathbb{S}^{\dimX-1}$.
In \cref{eq:setting}, we then sample instrument components independently from a Rademacher distribution, and set $\epsilon_X = \M e$, $\epsilon_Y = v^T e$ with all components of $e \in \bR^{\dimX}$ sampled independently from standard Gaussians for each sample.
This confounding model satisfies the assumptions required to estimate $\betahat{}$ \citep{janzing2018detecting,janzing2019causal,rendsburg2022consistent}.
Loosely speaking, the confounder affects both the treatment and the outcome as (independent) random mixtures of the same independent noise sources.
We ensure in our data generation that $\beta \in \im(\alpha)$.
Hence, $\beta$ can be recovered fully in principle.
Moreover, we also guarantee that there is a subset of less than $\nperround \cdot \nexperiments$ instruments that suffices to fully identify $\beta$ (see \cref{app:experiment_details} for details).
Therefore, even if we cannot use all instruments throughout the experiments ($\nperround \cdot \nexperiments < \ninstruments$) a good selection algorithm can in principle fully identify $\beta$.

For the similarities between instruments, in our experiments we take the absolute value of the cosine distance $\similarity_{i,j} = |\alpha_i^T \alpha_j| / (\|\alpha_i\|_2 \|\alpha_j\|_2)$.
Crucially, we thereby do not assume access to $\alpha$---only the similarities enter the algorithm.
While in a given application, the similarities would be provided by practitioners and domain experts, in our simulated study we have to choose \emph{some} similarity measure that is informed by $\alpha$.
We account for uncertainties in the provided similarities by computing them on a noisy version of $\alpha$, where we add independent standard Gaussian noise to all entries.
For the cost function we choose $c(d) = \log(d)$ for $d \le \nperround$ and $c(d) = \infty$ otherwise, effectively limiting the maximum number of instruments per round to $\nperround$.
While pairwise similarities can help the sequential selection to converge quickly, we note that our findings from \cref{sec:underidentified} are extremely useful for the (often strong) baseline of selecting instruments randomly.

\begin{figure}
    \centering
    \includegraphics[width=0.95\linewidth]{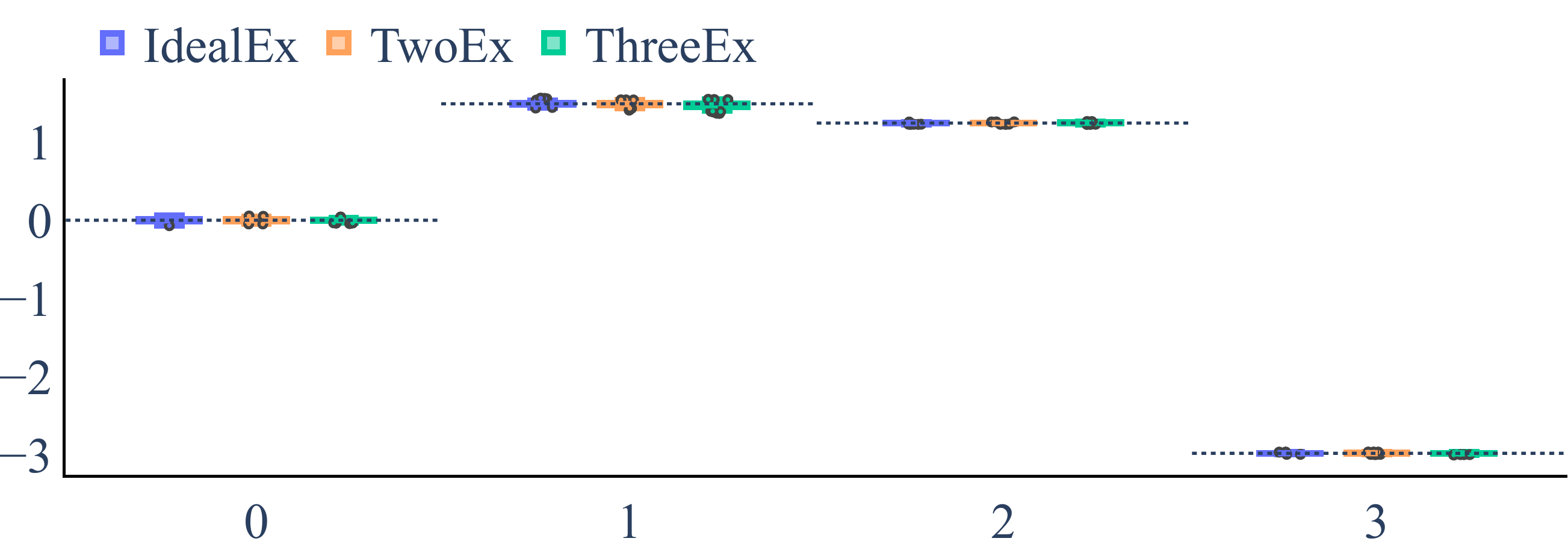} 
    \includegraphics[width=0.95\linewidth]{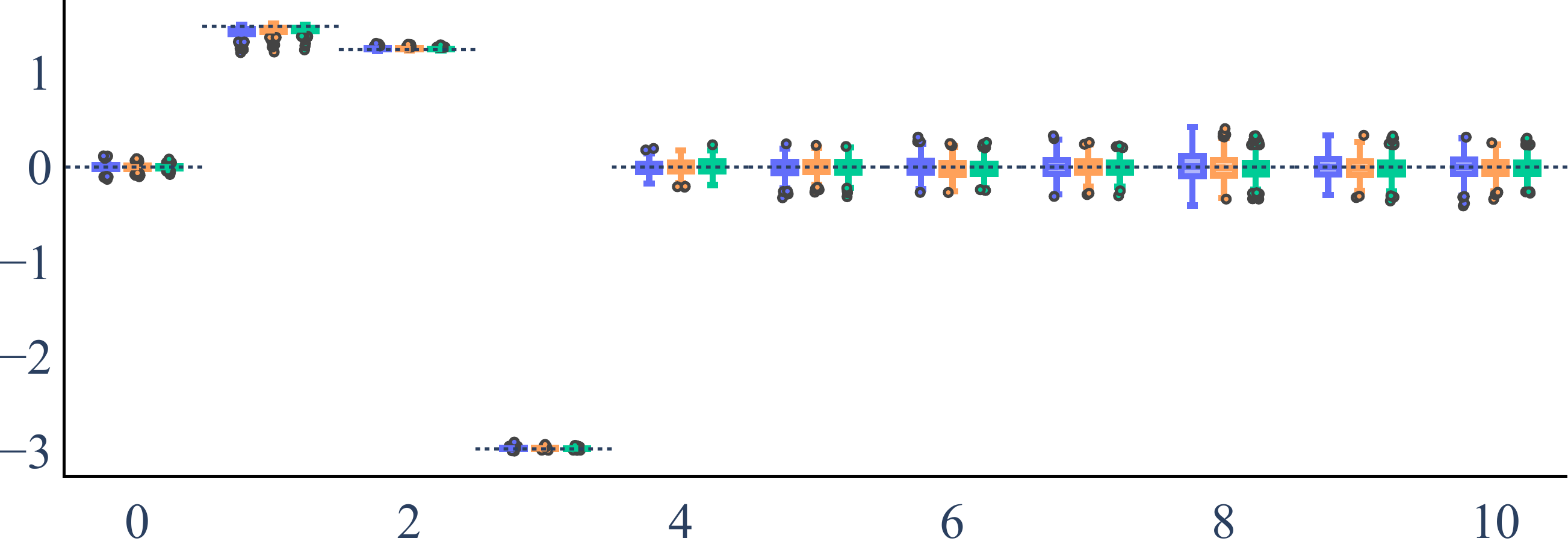} 
    \vspace{-3mm}
    \caption{Estimates of $\beta$ over 500 runs with $\dimZ = \dimX = 3$ (top) and $\dimZ=3$ and $ \dimX = 10$ (bottom). The dotted lines are the true $\beta$ and boxplots show the median (horizontal line), first and third quartile (box height) and 10/90 percentiles (whiskers).
    }\label{fig:finitesample}
    \vspace{-5mm}
\end{figure}

\begin{figure*}
    \centering
    \includegraphics[width=.48\linewidth]{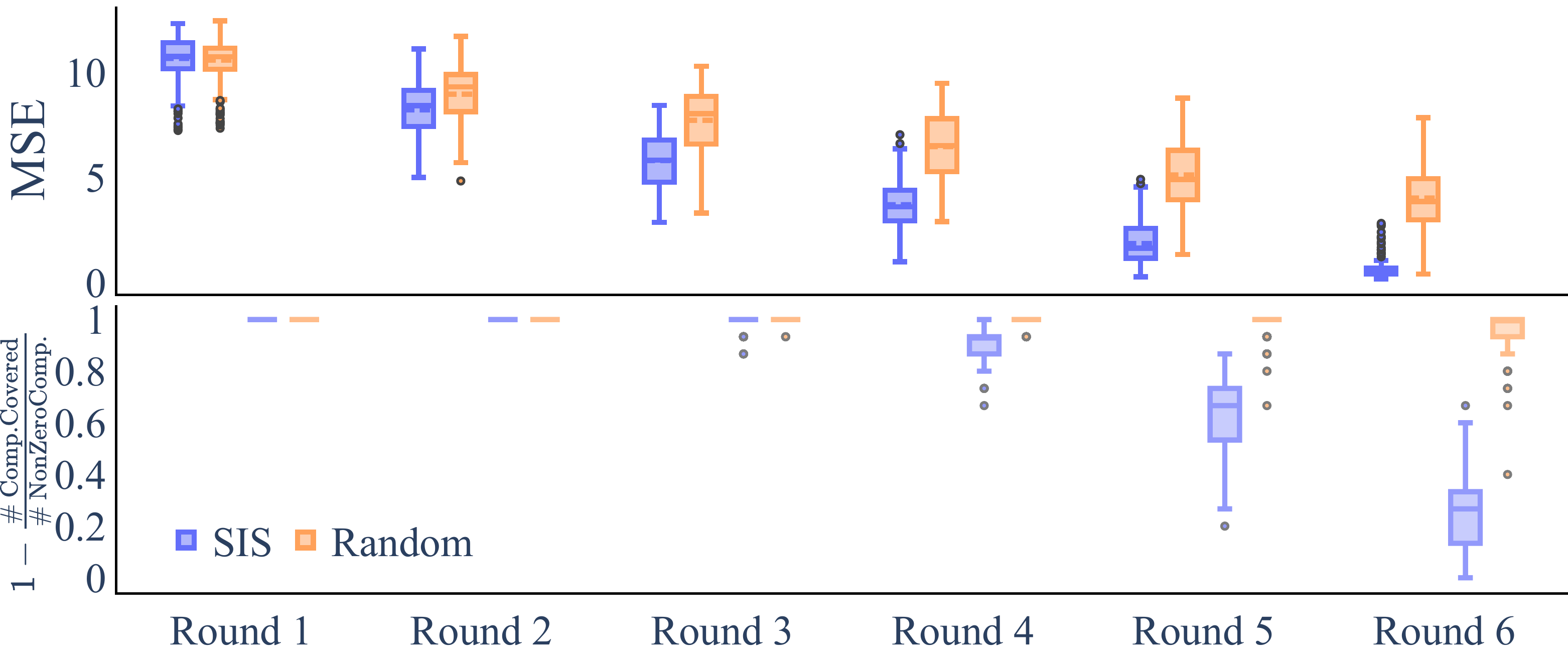} 
    \hfill
    \includegraphics[width=.48\linewidth]{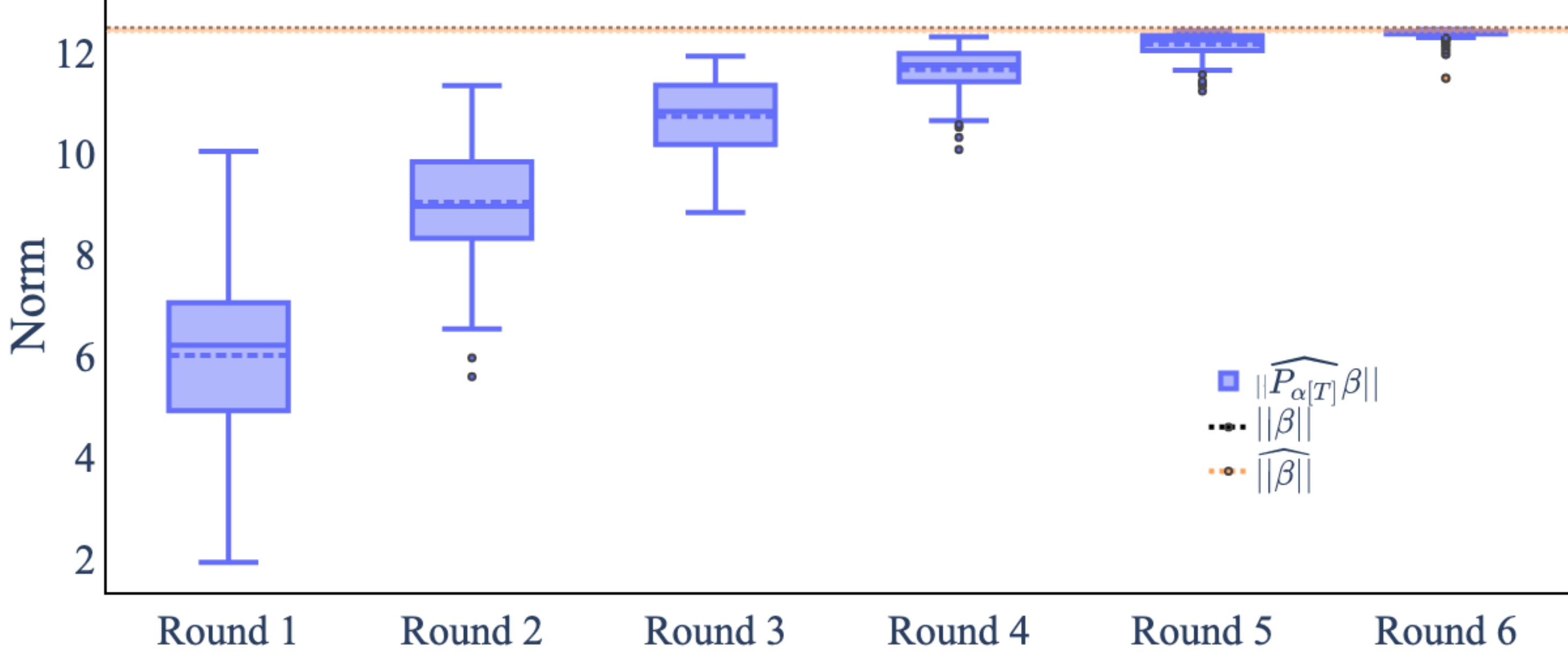} 
    \vspace{-3mm}
    \caption{Results for the sequential selection of instruments with $\dimZ=30, d_{\text{id}}=15, \nperround = 3, \nexperiments=6$ and $\dimX=50$: each boxplot shows the median and mean (solid resp.\ dashed line), first and third quartile (box height) and 10/90 percentiles (whiskers) over $n_\text{runs}=250$. \textbf{Left:} The upper panel shows the of squared error $\betaivsub{[t]} - \beta$ over rounds $t \in \{1, \dots, 6\}$. The lower panel shows the corresponding percentages of unidentified components.
    \textbf{Right:} The $||\widehat{\P_{\alpha[t]}\beta}||$ increases for $t \in \{1, \dots, 6\}$ approaching $\widehat{||\beta||}$ (dotted orange line), which perfectly estimates the true $\|\beta\|$ in this case (dotted black line).} \label{fig:seq_opt_results}
\end{figure*}

\begin{figure*}
    \centering
    \includegraphics[width=0.98\linewidth]{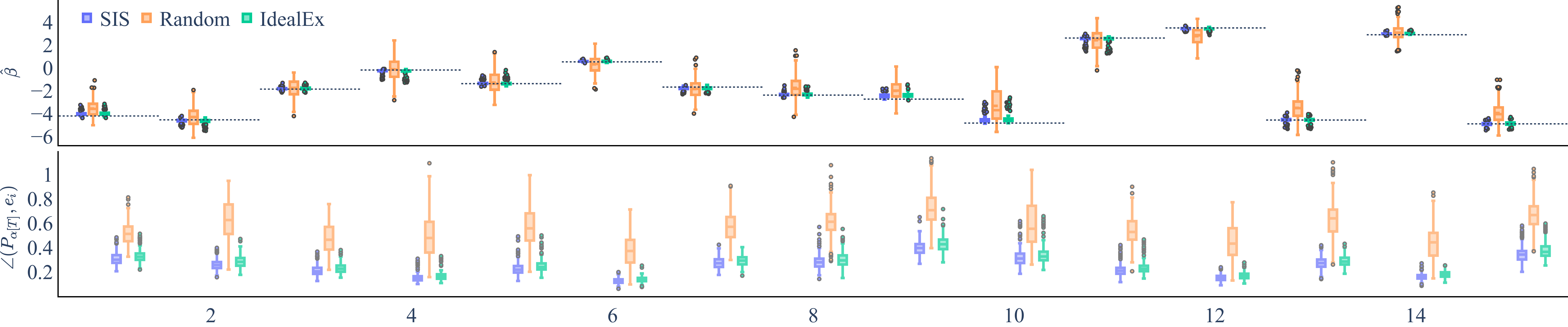} 
    \vspace{-3mm}
    \caption{Results for the sequential selection of instruments with $\dimZ=30, d_{\text{id}}=15, \nperround = 3, \nexperiments=6$ and $\dimX=50$: each boxplot shows the median and mean (solid resp. dotted line), first and third quartile (box height) and 10/90 percentiles (whiskers) over $n_\text{runs}=250$ after the last round $T=6$. \textbf{Top:} estimates of the nonzero components of $\beta$ with the black dotted line being ground truth. \textbf{Bottom:} distribution of cosine similarities in \cref{eq:identifiedcomps} for each component.} \label{fig:lastround_opt_results}
    \vspace{-2mm}
\end{figure*}

\xhdr{Finite sample properties of our estimators}
We first empirically analyze the finite sample behavior of our estimators from \cref{prop:identifyprojection,prop:combinedestimate}.
For illustration, \cref{fig:finitesample} shows a setting with $\dimX=\dimZ=3$ on the upper panel and $\dimX=10$, $\dimZ=3$ on the lower panel.
The x-axis shows component indices of $\beta \in \bR^{\dimX}$ (component 0 is the offset); dotted lines are the ground truth values of $\beta$, and boxplots show distribution of estimated values ($\betaivsub{}$) over 500 random seeds.
We compare three different estimators.
\emph{IdealEx} randomizes all three available instruments simultaneously and shows the single estimate (from \cref{prop:identifyprojection}).
\emph{TwoEx} combines (via \cref{prop:combinedestimate}) two individual estimates using only the first two and the last instrument, respectively.
\emph{ThreeEx} combines three individuals estimates, obtained by randomizing each of them separately.
\Cref{fig:finitesample} shows that all methods correctly identify all three components of $\beta$ (plus the offset) on average with low variance in the low-dimensional $\dimX=3$ setting.
The $\dimX=10$ setting (lower panel) in which 3 instruments suffice for full identification, shows similar results for all components.
Even though the variance increases compared to the lower-dimensional setting, estimates are still consistent.
All details are provided in \cref{app:experiment_details}.

\xhdr{Sequential instrument selection}
For our instrument selection algorithm we compare the following methods.
\emph{IdealEx}: the hypothetical ideal where all instruments are used at once in a single experiment. \emph{Random}: a baseline that selects one of the allowed (smaller or equal $\nperround$) subsets uniformly at random from the remaining instruments at each round. \emph{Sequential Instrument Selector (SIS)}: our \cref{alg:cap}. 
We remark that the random baseline, to its advantage, does not respect the cost per experiment, but is allowed to select the maximum number of possible instruments in each round.

We set $\dimZ=30, d_{\text{id}}=15$ and use two different treatment dimensions $\dimX=50$ and $\dimX=150$.
Further, we allow $\nperround = 3$ instruments per round with a total budget of $\nexperiments = 6$ experimental rounds.
Note that $\nperround \cdot \nexperiments \leq d_{\text{id}}$, i.e., the algorithm has the budget to fully identify $\beta$.

In \cref{fig:seq_opt_results} (\textit{left}) we show boxplots of the mean squared error (MSE) of our estimates for the nonzero components after each optimization round.
\textit{SIS} outperforms the random baseline in terms of MSE.
Additionally, each boxplot in the lower panel shows the percentage of uncertain components from \cref{eq:identifiedcomps} for $\delta = 0.3$.
In \cref{fig:seq_opt_results} (\textit{right}), the norm of our estimator $||\betaivsub{}||$ converges towards the norm of the actual $\beta$ (and its estimate $\betahat{}$), i.e., the stopping criterion is reached.
\Cref{fig:lastround_opt_results} compares these estimates for the nonzero components of $\beta$ after the last round of experiments.
Our greedy optimization indeed identifies each of the components just as well as the hypothetical ideal of a single experiment that uses all instruments at once.
Empirically, the finite sample properties remain unaffected by our combination procedure.
We refer to \cref{app:experiment_additional} for results on a $\dimX=150$ setting.

\section{Conclusion}

In this work we made multiple contributions aiming at inferring causal effects of high-dimensional treatments under unobserved confounding by sequential experimentation, where we cannot intervene on the treatments directly, but can only randomize instruments.
We proposed consistent, asymptotically normal estimators for the orthogonal projection of a treatment effect onto the instrumented subspace in the linear setting and introduced a method to consistently combine such estimates from separate experiments.
Surprisingly, neither the (perhaps intuitive) estimator in \cref{prop:identifyprojection}, nor the geometric intuition around instrumented subspaces in the underspecified setting can be found in existing literature.
These estimators may be of independent interest as a contribution to the largely ignored underspecified IV setting.
We then developed an algorithm to sequentially propose subsets of instruments from a given pool that flexibly trades off the expected information gain (informed by provided similarities) with the cost of each experiment.
Moreover, we integrated a stopping criterion for when the sequential selection has fully identified the causal effect, a method to keep track of all components that are consistently estimated, and an upper bound on the absolute error of unidentified components: these additions inform the practitioner about whether (and which parts) of the estimate can be trusted.

The linearity assumption may appear restrictive.
However, the linear IV setting is still heavily used in econometrics and health, as it reliably captures dominant effects even in noisy settings, and still attracts attention with novel results recently \citep{pfister2022identifiability,rothenhausler2018anchor}.
The thorough theoretical understanding developed in this work is a challenging and necessary foundation for experiment design via instrument selection.
Extensions of our method and of our notion of the instrumented subspace to (certain) non-linear settings is an important direction for future work.
A second limitation of our work is inherent to the problem setting: missing real-world experiments due to a lack of access to the required expensive, specialized facilities.
Finally, we highlight that independent testing and verification is paramount when using algorithmically obtained insights to inform consequential decisions such as actual clinical treatment decisions.

\xhdr{Code}
The implementation as well as experimental details are publicly available on Github: 
\href{https://github.com/EAiler/underspecified-iv}{https://github.com/EAiler/underspecified-iv}.

\xhdr{Acknowledgments}
We would like to thank Sören Becker for useful suggestions and Dominik Janzing for fruitful discussions on estimating $\|\beta\|_2$.
EA is supported by the Helmholtz Association under the joint research school ``Munich School for Data Science - MUDS''.
JH is supported by the Natural Sciences and Engineering Research Council of Canada (NSERC) and Recursion Pharmaceuticals.
This work has been supported by the Helmholtz Association’s Initiative and Networking Fund through CausalCellDynamics (grant \# Interlabs-0029).

\bibliography{bib}
\bibliographystyle{icml2023}

\newpage
\appendix
\onecolumn
\newpage

\section{Details of Experiments} \label{app:experiment_details}

\paragraph{Data Generation}

The data generation process follows the model described in \cite{janzing2018detecting}.
We adopt their data setup in order to estimate $\widehat{||\beta||}$ consistently:
\begin{equation}
  X = Z \alpha + \epsilon_X \:, \qquad
  Y = X \beta + \epsilon_Y\:,
\end{equation}
$\epsilon_X$ and $\epsilon_Y$ are confounded via the common variable $e$.
\begin{equation}
  \epsilon_X = eM \:, \qquad
  \epsilon_Y = ev\:,
\end{equation}
with $Z \sim \text{Rademacher}(0.5)$ and $e \sim \mathcal{N}(0, Id_l)$.

\paragraph{Scenario Generation.}

For each figure, i.e. simulation, we generate random scenarios in order to avoid introducing involuntary bias in the parameter setting.

Each generated scenario is based on a seed and the constants $n, \dimX, \dimZ$ and $d_\text{id}$, i.e. the number of instruments it will take to identify the causal effect in full.
Further, we assume $\dimX=l$, i.e. $M \in \bR^{\dimX \times \dimX}$. Note that this choice is based on the setting in the simulation studies of \cite{janzing2018detecting}.

We sample $d_\text{id}$ nonzero components of $\alpha_j, j \in \{1, ..., d\}$ and the $d_\text{id}$ nonzero components of $\beta$ from a uniformly distributed random variable $\mathcal{U}(-5.0, 5.0)$.
The sparsity is introduced by setting the remaining components to $0.0$. 
Our motivation for this choice is to reliably generate a setup for which we are guaranteed to identify the causal effect by a maximum of $d_\text{id}$ instruments.
In addition to $d_\text{id}$ identifying instruments, we generate $d-d_\text{id}$ further instruments by picking two of the necessary instruments on top of which we add a $p$-dimensional Gaussian noise.
Overall, we end up with $d_\text{id} - 2$ necessary instruments and two clusters from which we can pick any instrument in order to identify the remaining components of $\beta$.
For the confounder we sample all entries of $\M \in \bR^{\dimX \times \dimX}$ from independent standard Gaussians and the direction $v \in \bR^{\dimX}$ as a uniform sample from the sphere $\mathbb{S}^{\dimX-1}$.

\paragraph{Scenario Generation for \cref{fig:finitesample}.}
For showcasing the finite sample properties of our estimator, we sampled from the scenario generation above with seed $253$ and $\dimZ=3$ and $\dimX=3$ resp. $\dimX=10$.
This choice of parameters left us with two settings (1) being just-identified $\dimZ=\dimX = 3$ and (2) being underspecified $\dimZ = 3, \dimX = 10$.
Moreover, we set $\beta$ to only have $3$ non-zero components in order to be able to identify the causal effect in full. Note that this choice is for illustration purpose and does not affect the method's applicability in a setting where we can only identify parts of the causal effect. However, in those setting where some $\beta_i$-components are not part of the instrumented subspace $\im(\alpha)$, full causal recovery is in general impossible.

\section{Additional Experiments} \label{app:experiment_additional}

In the optimization we compare the baseline which uses all instruments at once (\textit{IdealEx}) to the developed sequential optimization routine (\textit{SIS}) and the random baseline.
We include results for the same setting as \cref{fig:seq_opt_results} and \cref{fig:lastround_opt_results}, except with an increased treatment dimension, i.e. $\dimX=150$ in \cref{appfig:opt_results_1}:
\begin{equation*}
    \dimX = 150, \;
    \dimZ = 30, \;
    d_\text{id} = 15, \;
    \nperround = 3, \;
    \nexperiments = 6 
\end{equation*}

\begin{figure*}
    \centering
    \includegraphics[width=.48\linewidth]{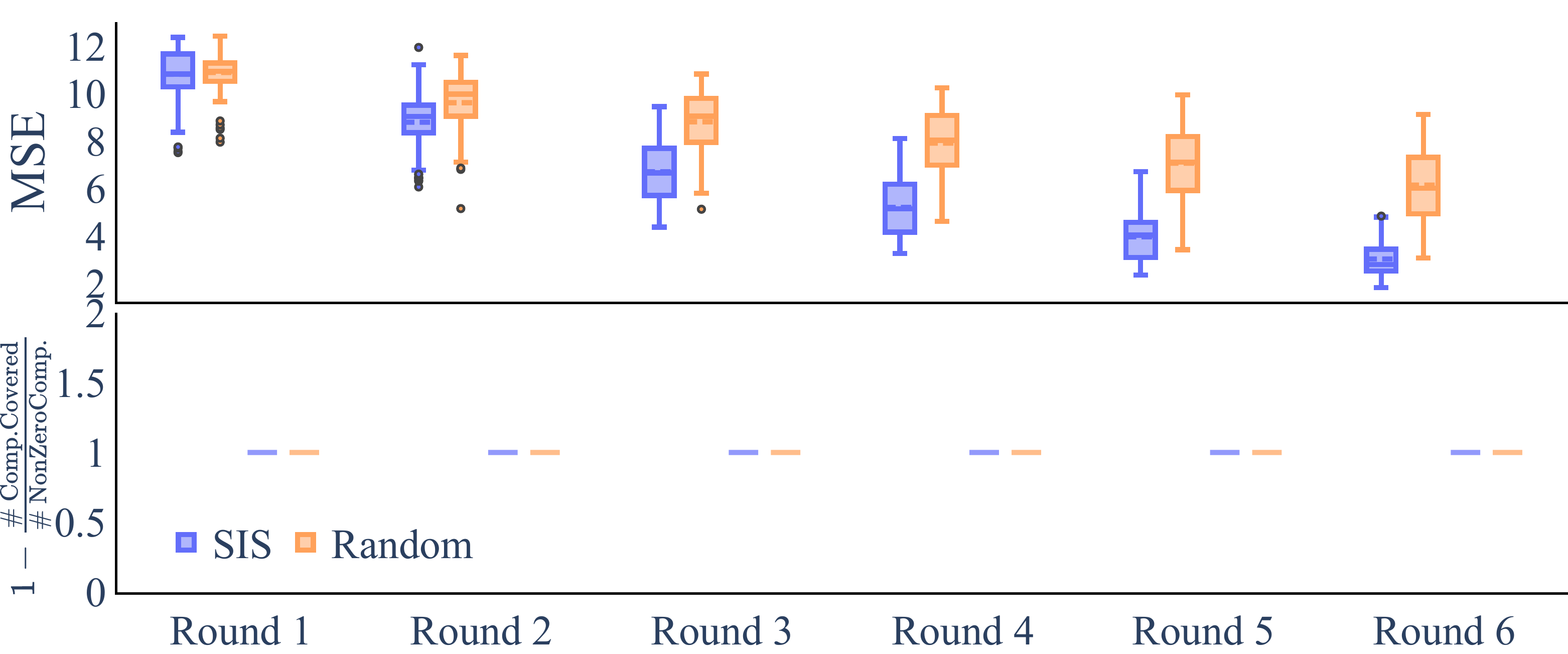} 
    \hfill
    \includegraphics[width=.48\linewidth]{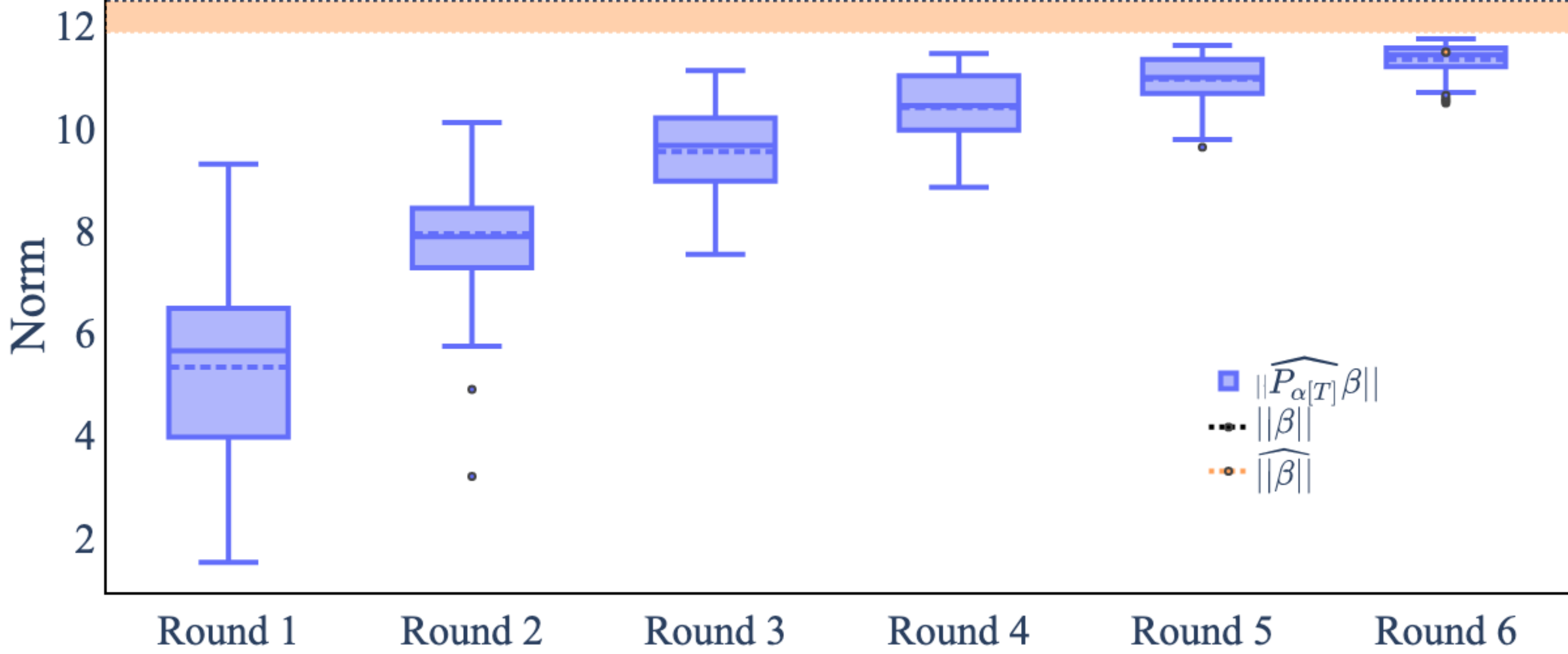}
    \includegraphics[width=0.98\linewidth]{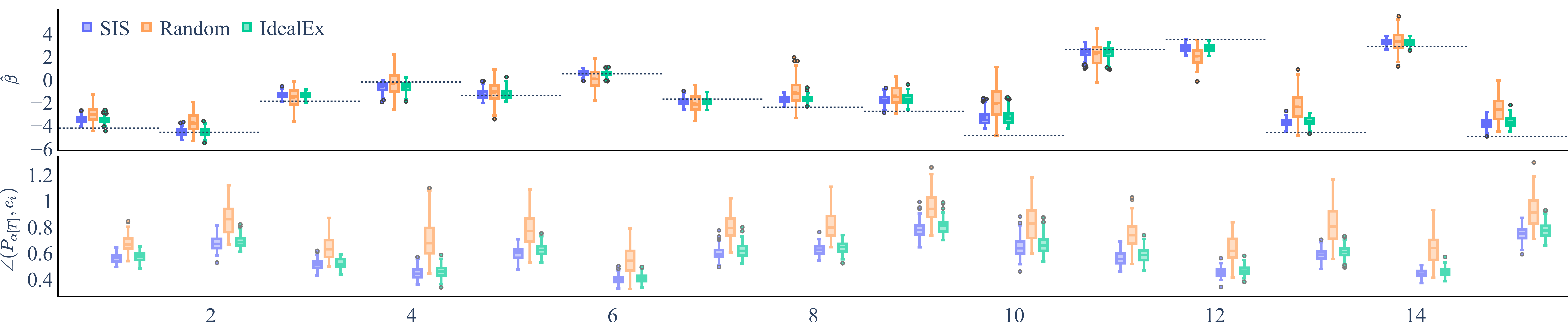} 
    \vspace{-3mm}
    
    \caption{Results for the sequential selection of instruments with $\dimZ=30, d_{\text{id}}=15, \nperround = 3, \nexperiments=6$ and $\dimX=150$: each boxplot shows the median and mean (solid resp.\ dashed line), first and third quartile (box height) and 10/90 percentiles (whiskers) over $n_\text{runs}=250$. \textbf{Top Left:} The upper panel shows the of squared error $\betaivsub{[t]} - \beta$ over rounds $t \in \{1, \dots, 6\}$. The lower panel shows the corresponding percentages of unidentified components. As we are in the scenario with $p=150$, we would need to adjust the threshold $\delta=0.3$ to a higher value.
    \textbf{Top Right:} The $||\widehat{\P_{\alpha[t]}\beta}||$ increases for $t \in \{1, \dots, 6\}$ approaching $\widehat{||\beta||}$ (shaded orange block), which does not perfectly estimate the true $\|\beta\|$ in this case (dotted black line), but performs still reasonably well.
    \textbf{Bottom:} \textit{Upper Panel:} estimates of the nonzero components of $\beta$ with the black dotted line being ground truth. \textit{Lower Panel:} distribution of cosine similarities in \cref{eq:identifiedcomps} for each component.
    } \label{appfig:opt_results_1}
\end{figure*}

Moreover, it might not always be the case that the stopping criterion in \cite{janzing2018detecting} works as nicely as in the previous scenarios, see Fig.~\cref{appfig:opt_results_2} with parameter setting:
\begin{equation*}
    \dimX = 50, \;
    \dimZ = 30, \;
    d_\text{id} = 20, \;
    \nperround = 4, \;
    \nexperiments = 6
\end{equation*}

\begin{figure*}
    \centering
    \includegraphics[width=.48\linewidth]{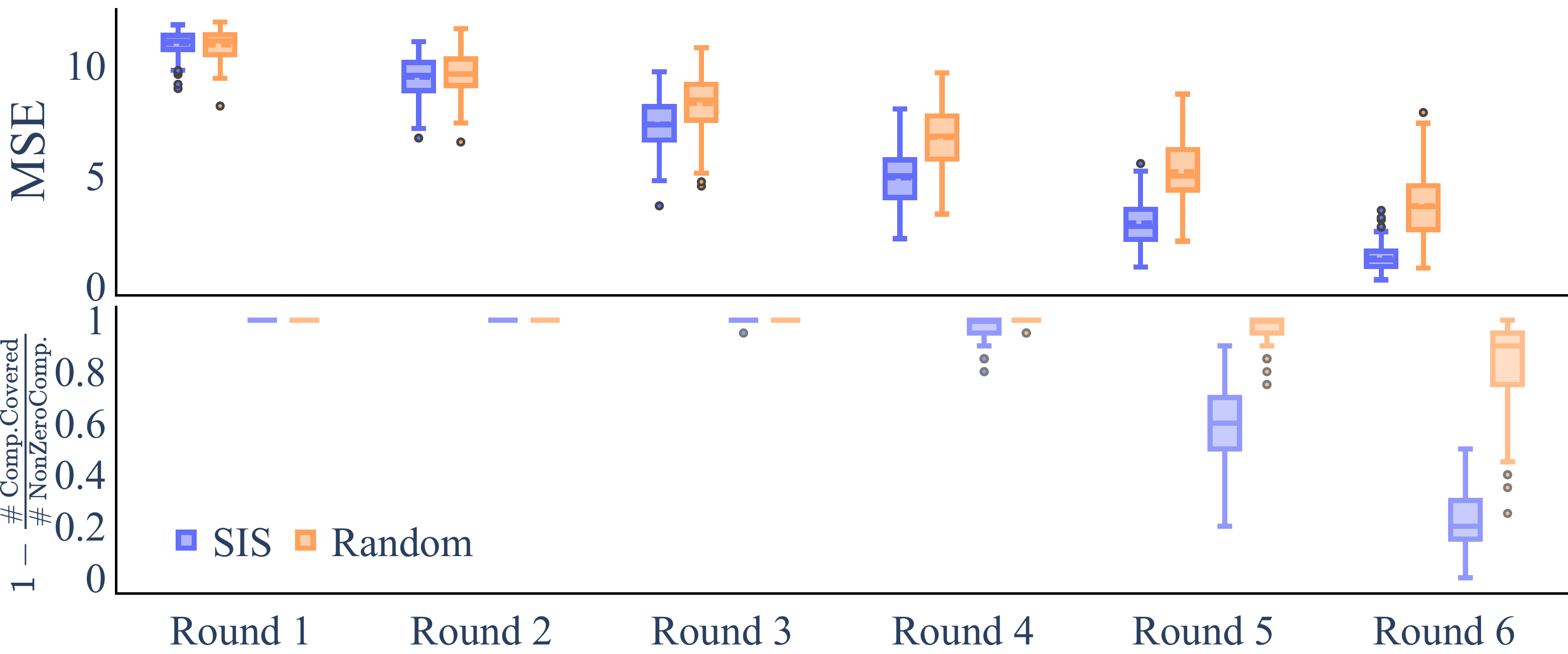} 
    \hfill
    \includegraphics[width=.48\linewidth]{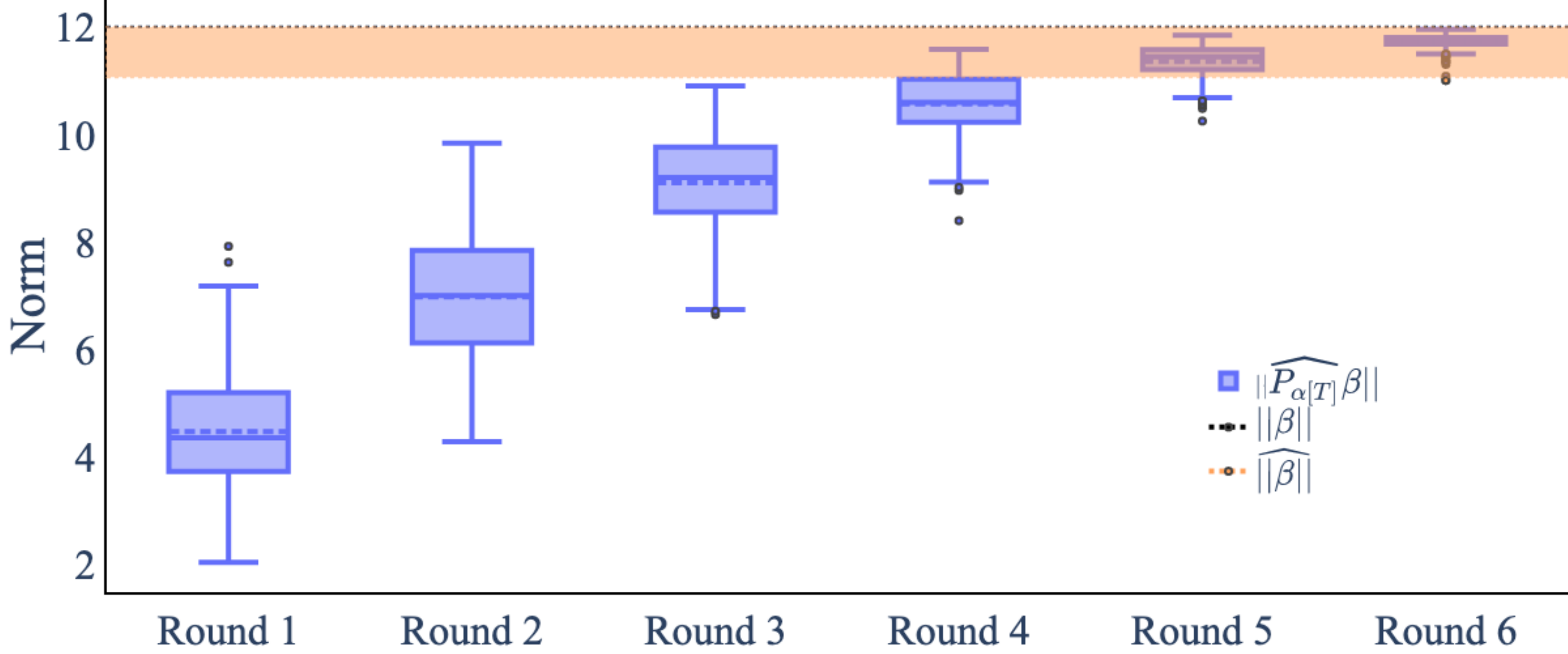}
    \includegraphics[width=0.98\linewidth]{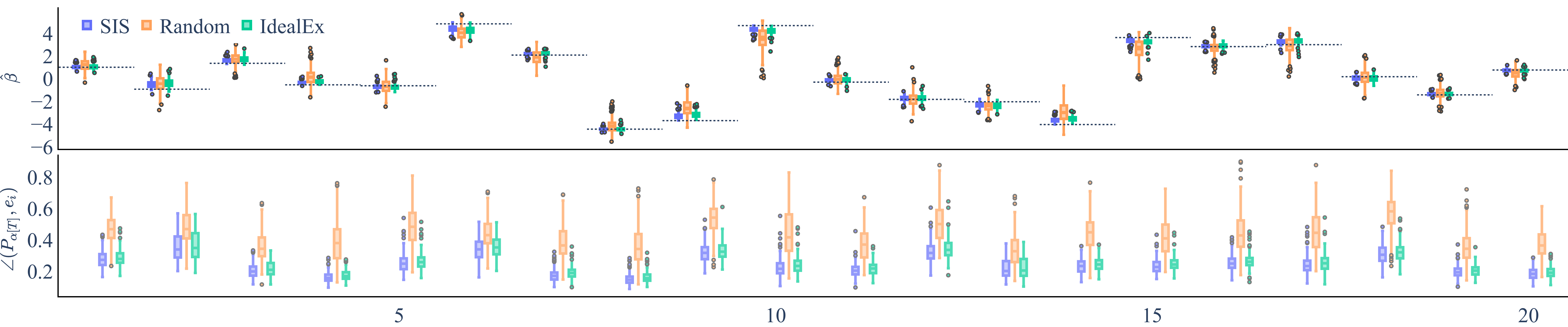} 
    \vspace{-3mm}
    
    \caption{Results for the sequential selection of instruments with $\dimZ=30, d_{\text{id}}=20, \nperround = 4, \nexperiments=6$ and $\dimX=50$: each boxplot shows the median and mean (solid resp.\ dashed line), first and third quartile (box height) and 10/90 percentiles (whiskers) over $n_\text{runs}=250$. \textbf{Top Left:} The upper panel shows the of squared error $\betaivsub{[t]} - \beta$ over rounds $t \in \{1, \dots, 6\}$. The lower panel shows the corresponding percentages of unidentified components.
    \textbf{Top Right:} The $||\widehat{\P_{\alpha[t]}\beta}||$ increases for $t \in \{1, \dots, 6\}$ approaching $\widehat{||\beta||}$ (shaded orange block), which does not perfectly estimate the true $\|\beta\|$ in this case (dotted black line), but performs still reasonably well.
    \textbf{Bottom:} \textit{Upper Panel:} estimates of the nonzero components of $\beta$ with the black dotted line being ground truth. \textit{Lower Panel:} distribution of cosine similarities in \cref{eq:identifiedcomps} for each component.
    } \label{appfig:opt_results_2}
\end{figure*}

\end{document}